\documentclass[journal]{IEEEtran}

\usepackage{subfig}
\usepackage{booktabs}

\usepackage{hyperref}

\usepackage{amssymb}
\usepackage{graphicx}
\usepackage{amsmath}
\newtheorem{theorem}{Theorem}
\newtheorem{proof}{Proof}
\newtheorem{definition}{Definition}

\newtheorem{lemma}{Lemma}

\usepackage{etoolbox}

\begin{document}
\title{Technical Report: Scheduling Flows with Multiple Service Frequency Constraints}
\author{Xu~Zheng,
        Zhipeng~Cai,~\IEEEmembership{Senior Member,~IEEE,}
        Jianzhong~Li,
        and~Hong~Gao

\thanks{X. Zheng, J. Li and H. Gao are with the School of Computer Science,
Harbin Institute of Technology, Harbin, 150001, China.
Email: \{zhengxu, lijzh, honggao\}@hit.edu.cn}

\thanks{Z. Cai is with the Department of Computer Science,
Georgia State University, 25 Park Place, Atlanta, GA 30303, USA.
Email: zcai@gsu.edu}}

\maketitle

\begin{abstract}

With the fast development of wireless technologies, wireless applications have invaded various areas in people's lives with a wide range of capabilities.
Guaranteeing Quality-of-Service (QoS) is the key to the success of those applications.
One of the QoS requirements, service frequency, is very important for tasks including multimedia transmission in the Internet of Things.
A service frequency constraint denotes the length of time period during which a link can transmit at least once.
Unfortunately, it has not been well addressed yet.
Therefore, this paper proposes a new framework to schedule multi transmitting flows in wireless networks considering service frequency constraint for each link.
In our model, the constraints for flows are heterogeneous due to the diversity of users' behaviors.
We first introduce a new definition for network stability with service frequency constraints and demonstrate that the novel scheduling policy is throughput-optimal in one fundamental category of network models.
After that, we discuss the performance of a wireless network with service frequency constraints from the views of capacity region and total queue length.
Finally, a series of evaluations indicate the proposed scheduling policy can guarantee service frequency and achieve a good performance on the aspect of
queue length of each flow.

\end{abstract}

\section{Introduction}
During the past decade, the diversity of wireless network architectures manifest a tremendous growth with the improvements of mobile devices and wireless communication techniques.
Novel systems like Internet of Things (IoTs) or Cyber Physical Systems (CPSs) \cite{atzori2010internet} have changed a wireless network from a simple data transmitting platform to a complicated one.
Such systems often request the support of real-time multimedia data transmission like Video on demand (VoD) and voice over IP (VoIP) \cite{chen2007applications},
which pose various Quality-of-Service (QoS) requirements for including packet delay constraints, delivery ratio, traffic admission control, link service regularity, \emph{etc.}
for link scheduling in wireless networks.
Motivated by this, a series of works have emerged which focus on scheduling traffics considering QoS constraints.
Many of them consider the transmission of packets with stringent delay constraints \cite{11jaramillo2011scheduling},
while some others focus on delivery ratio \cite{2013deficitkang2013performance}, admission control, \emph{etc}.
Surprisingly, service regularity, one of the most important QoS requirements for many streaming multi-media applications, has not been fully explored.
To the best of our knowledge, only few works \cite{13interservicetimeli2013throughput} \cite{13heavytailli2013heavy} concentrate on the link scheduling problem considering the
service regularity requirement.


Service regularity refers to the time between two consecutive services of a link.
It indicates a link will not wait for too much time for the next service and the waiting time between every consecutive transmission pair is relatively identical.
In traditional data-oriented wireless networks, a link has no constraint on service regularity as long as its overall throughput is guaranteed.
Such tolerance no longer exists in the current IoTs carrying out multimedia data.
For example, in some applications, videos are uploaded by the citizens to report emergencies around them.
The officers at the controlling center request the videos to be presented smoothly to support accurate and rapid decisions.
It indicates that the video blocks from each uploader should be transmitted periodically to guarantee the continuous playback of the video.
This puts forward a service regularity requirement for the scheduling policy.
Li {\it et al.} \cite{13interservicetimeli2013throughput} introduced a new parameter called \emph{time-since-last-service} to improve service regularity.
The time-since-last-service indirectly refers to the time between two consecutive transmissions of a link.
Their work balances the long-term average inter-service time of each user while still maintains the stability of the network.
However, it fails in the situation when the service regularity of each link is strictly bounded.
In this paper, we fill this gap by designing a novel scheduling policy.  
Within our framework, service frequency requirement, a variant of service regularity, of all the links can be completely satisfied.

Service frequency refers to the length of time during which a link obtains a chance to transmit.
A constraint on service frequency indicates that a link should be scheduled at least once in every frame if we partition time into frames with length equaling to the constraint.
Service frequency is closely related to service regularity in that it can bound the waiting time between two consecutive transmissions of a link, and can also lead to
a stable waiting time if the links are properly scheduled.
Service frequency is also essential for multimedia data transmission.
Again, notice that the smoothly real-time playback of videos or voices are very important for some crowdsourcing monitoring systems,
and such multimedia data is played in a batched manner \cite{liu2013user}.
An indispensable QoS requirement is to deliver the next block to the user before the end of her current one.
\emph{i.e.}, the user must get a new block in each frame whose length equals to the playing time of a block.

Furthermore, the service frequency requirements of users are often heterogeneous due to the diversity of users' behaviors.
For example, users may upload various types of multimedia data or videos with different resolutions.
Higher quality indicates larger service frequency when the size of a packet is fixed.
To the best of knowledge, this is the first work to investigate the scheduling of flows in wireless networks with heterogeneous service frequency constraints.

Service frequency requirements in fact make a scheduling policy a combination of the round-robin and maximum-throughput methods.
Users, on the one side, upload their multimedia data to participate in the monitoring tasks.
This is like the round-robin scheduling, where each user is served periodically.
At the same time, all users also hope their own applications running normally on mobile devices, which means the ones with longer packet queue must have higher priority to be served.
This necessitates maximum-throughput scheduling.
In this paper, we introduce a new definition for capacity region with service frequency constraints.
The novel capacity region is the set of packet-arriving rate vectors in which no user will suffer from buffer overflowing and all the constraints can be satisfied.
This new definition can also be utilized to judge the performance of a scheduling policy.
Furthermore, we propose a novel scheduling policy that considers service frequency constraints, and the policy is throughput-optimal in some types of networks.
Our main contributions are summarized as follows:
\begin{itemize}
\item We propose a novel network model accommodating service frequency constraints, which is more suitable for many real-time wireless applications.
We also provide new definitions of the stability of a system and the capacity region.
\item We design a new scheduling policy.
The policy is very easy to be implemented.
Furthermore, we prove that the policy is throughput-optimal in an extensively utilized category of network model. 
\item We compare the size of the new capacity region with the original one without the service frequency constraints.
Meanwhile, we analyze the traffic load of in a wireless network when a system is stable.
\item Simulation results reveal that our scheduling policy achieves a good performance in both queue length and service regularity.
\end{itemize}

The rest of the paper is organized as follows.
In Section 2, we have a literature review on the development of scheduling policies in wireless networks.
Section 3 introduces the network model.
In Section 4, the novel scheduling policy is proposed with some theoretical analysis.
Section 5 presents the performance of a wireless network with service frequency constraints.
The simulation results of our scheduling policy are illustrated in Section 6.
Section 7 concludes the paper.

\section{Related Works}
Scheduling is a fundamental issue for wireless networks.
It is well-known that the Maximum Weight scheduling policy is throughput-optimal \cite{27tassiulas1992stability}.
The weight is usually denoted as the ratio between queue length and transmitting rate.
Gupta and Shroff \cite{gupta2010delay} showed that a generalized maximum weight scheduling could also be throughput-optimal.
The policy scaled the queue length of each link by a weighted factor.
Some work \cite{boyaci2014optimal} considers queue length and delay under the max weight scheduling policy.

However, the problem of finding a disjoint set of links with the maximum weight under the general interference model is NP-Hard.
Hence, a maximum weight scheduling policy incurs high overhead.
Much effort has been spent on designing low complexity scheduling policies with throughput performance guarantees.
Some policies are based on the \textit{picking-and-comparing method} \cite{6pickcomparebui2009distributed} or \textit{Carrier Sensing Multiple Access (CSMA)} \cite{23CSMAni2012q}.
Those policies gradually achieve the optimal throughput after many iterations, but are insensitive to channel fading \cite{19opporscheliu2003framework}.
Lin and Shroff \cite{lin2006impact} proved that the maximum weight scheduling can achieve a guaranteed ratio of the capacity region.
The performance is lower bounded by the ratio bound of the weight of the maximal independent set.
Since then, many scheduling policies focus on finding a \textit{maximal weight independent set (MWIS)} of links.
The works in \cite{24centralsakai2003note} and \cite{28centraldwan2011wireless} present some centralized policies that can find a MWIS with provable ratio bounds.
The works in \cite{4distributedbasagni2001finding}, \cite{10distributedhoepman2004simple} \cite{13distributedjoo2013distributed} propose some distributed scheduling policies for different network models.

However, none of the aforementioned works deal with the QoS requirements of users.
Some scheduling policies are designed to meet the packet delay constraint for each user \cite{9hou2010utility}\cite{10jaramillo2010optimal}\cite{11jaramillo2011scheduling}.
They all partition time into frames, and packets will run out of time before the end of their arriving frames.
They mainly differ in the ways of dealing with channel fading \cite{8hou2010scheduling}, nonidentical delays \cite{11jaramillo2011scheduling}, \emph{etc}.
The work in \cite{2013deficitkang2013performance} proposes a new policy that accommodates packet delivery ratio.

Very few works \cite{13interservicetimeli2013throughput} \cite{13heavytailli2013heavy} deal with service regularity.
The work in \cite{13interservicetimeli2013throughput} introduces a new type of parameter to capture the performance of service regularity, which is called time-since-last-service.
It enriches the definition of weight to include the changing of this parameter.
The weight is a combination of queue length and \textit{time-since-last-service}.
Then under the maximum weight scheduling policy, a user will get a second service with a growing probability as \textit{time-since-last-service} increases, thus the inter-service time can be regulated.
Li \emph{et al}. \cite{13heavytailli2013heavy} extended their work to consider the inter-service time performance of rate vectors close to the edge of capacity region.
This work could still achieve the original capacity region, and only relatively improve the performance of service regularity.
Actually, this work can achieve an approximate guarantee on service regularity.
This is done by carefully choosing the parameters according to a comprehensive knowledge of users' behaviors, which is impossible in practice.
It would fail when the service regularity of each user is exactly bounded and the users' behaviors change often.
Besides, it cannot figure out the new capacity region in which the users follow strict service frequency requirements and the system is stable. 
This motivates us to design a novel scheduling policy which can satisfy the heterogeneous service frequency of each user.

\section{Network Model and Preliminaries}
\subsection{Network Model}
A network is composed of $N$ links, denoted as $\{l_1,l_2,\cdots,l_N\}$.
They transmit via a common wireless channel.
Each link connects one source node to one destination node.
Due to the inherent interference in wireless networks, two links can not transmit simultaneously when they share the same node or are geographically close to each other.
The conflicts among links are presented by \textbf{\textit{conflict graph}}  $G_{NN}=\{g_{ij}\}_{i,j=1,2,\cdots,N}$ with binary entry.
$g_{ij}=1$ means link $l_i$ and $l_j$ conflict with each other and can not transmit simultaneously, and $g_{ij}=0$ otherwise.
Time is partitioned into slots, and in each slot a link can successfully transmit only when all the other links remain waiting.

In each slot, there are $A_i[t]$ packets arriving at link $i$.
$A_i[t]$ is randomly and identically distributed over time with expectation $E(A_i[t])=\lambda_i$.
The second moments of $A_i[t]$ are finite for all $i=1,2,\cdots,N$ and $t=1,2,\cdots$.
These packets will join the delivery queue of link $i$ and be delivered to the destination.
We assume that the wireless channel is unreliable.
The channel condition of link $l_i$ refers to the probability $c_i$ that $l_i$  can successfully transmitted in a slot. 
Then bandwidth for each link is represented as the number of delivered packets, \textit{i.e.}, link $l_i$ can deliver $r_i$ packets per slot when succeeded.
Finally, each link has a service frequency constraint $\delta_i$, which means link $i$ needs to be scheduled at least once in every frame with $\delta_i$ slots.
In our model, $\lambda_i$s, $c_i$s, $r_i$s, and $\delta_i$s are all heterogeneous, which indicates the fact that the traffic load, communication capability and QoS requirements of different mobile devices are often various.
In this paper, we consider the case when the service frequency constraints are fixed, and deal with different arriving rates.

We use boolean variable $s_i[t] \in \{0,1\}$ to denote whether user $i$ is scheduled in slot $t$, with $1$ for scheduling and $0$ for waiting.
The vector $S[t]=\{s_1[t],s_2[t],\cdots,s_N[t]\}$ is called a scheduling for slot $t$.
In our model, it is guaranteed that
\begin{center}
$\sum_{i=1}^{N}s_i[t]$
\end{center}
is maximal and all chosen links does not conflict with each other for every time slot $t$.

$Q_i[t]$ is the time consumption of transmitting all packets on link $i$ at the beginning of slot $t$.
$Q_i[t]$ is determined by $Q_i[t-1]$, $\frac{A_i[t]}{r_ic_i}$, and $s_i[t-1]$ as follows:
\begin{equation}
Q_i[t]=Max\{Q_i[t-1] - s_i[t-1],0\}+ \frac{A_i[t]}{r_ic_i}
\end{equation}
where $0$ means that user $i$ has delivered all its packets in slot $t-1$.
$Q_i[t]$ is closely related to the packet size on link $i$, while it can simplify our analysis.

Now we formalize the stability of networks by the following definitions.

\begin{definition}
A network is stable under an arriving rate vector $\lambda=\{\lambda_1,\lambda_2,\cdots,\lambda_N\}$ when

(1) the underlying Markov chain $Q_i[t]$ for each link is positive recurrent, and
\begin{equation}
\lim\limits_{t\rightarrow\infty}\sum_{i=1}^{N}E(Q_i[t])\leq \infty,
\end{equation}

(2) for every link $i$, it is scheduled at least once in every $\delta_i$ slots.
\end{definition}

We also use the term $\lim\limits_{t\rightarrow\infty}\frac{\sum_{t'=1}^{t}\sum_{i=1}^{N}Q_i[t']}{t}\leq \infty$ in this paper, which is actually the same as Formula (2).
A scheduling policy refers to a series of scheduling $S[1],S[2],S[3],\cdots$, one for each time slot.
An arriving rate vector $\lambda=\{\lambda_i\} ~(i=1,\cdots,N)$ is \textit{supportable} by a network if the network is stable under $\lambda$  when it utilizes some scheduling policies.
All supportable rate vectors compose the \textit{capacity region} of the network, denoted by $\Lambda$.
The throughput-optimality of a scheduling policy is defined as follow:

\begin{definition}
A scheduling policy is throughput-optimal for a network if $\Lambda'=\Lambda$, where $\Lambda'$ is the set of rate vectors supported by the network when using the scheduling policy.
\end{definition}

We will show in next section that to design a throughput optimal solution is NP-hard.
Our object is to design a scheduling policy that can achieve a good performance on throughput while maintaining a polynomial time complexity.

\subsection{Preliminaries}
Now we introduce some knowledge regarding the positive recurrent of Markov chains and the upper bounds on queue length of some well-known scheduling policies.
There exist a lot of methods which prove the positive recurrent of Markov chains.
In this paper, we use the Foster's criteria for positive recurrent and ergodic Markov chains \cite{leonardi2001stability}, which is shown as follows.
\begin{theorem}
A countable Markov chain $Q[t]=\sum_{i=1}^{N}Q_i[t]$ is positive recurrent and ergodic if and only if there exists a positive function $V>0$
and a finite set of state $\epsilon_0$, such that

(1) when $Q[t]\in \epsilon_0$, $\Delta(Q[t])<\infty$,

(2) when $Q[t]\notin \epsilon_0$, $\exists \xi > 0$, and $\Delta(Q[t])<-\xi$,

\noindent where
\begin{equation}
\Delta(Q[t])=E(V(Q[t+1])-V(Q[t])|Q[t]).
\end{equation}
$\blacksquare$

\end{theorem}

There also exist many results on the queue length of a network with the maximum weight scheduling policy.
Our work is not dominated by the choice of the result.
Thus, this paper uses a simple bound given in \cite{gupta2010delay}.

\begin{theorem}
Given a rate vector $\lambda=\{\lambda_i\} ~(i=1,\cdots,N)$ inside the capacity region, the following bound on the expectation of the sum of queue length is true in a system under the maximum weight scheduling policy:
\begin{center}
$\sum_{i=1}^{N}E(Q_i)\leq\frac{\lambda_i+Var[A_i]-\lambda^2_i}{2}$
\end{center}
\end{theorem}
$\blacksquare$

\section{Multi-Stage Maximum Weight Scheduling Policy}
In this section, we first prove the complexity of the problem and propose a new scheduling policy which is a combination of round-robin and  maximum-weight scheduling.
In our scheduling policy, links are first categorized into different stages indicating the priority level this link should be served. 
This mechanism can ensure the service frequency constraint of each link.
Meanwhile, a maximum weight scheduling policy is used for allocating the service among links of the same stage, resulting in a link with longer queue size to be served as soon as possible.
We call this framework \textit{Multi-Stage Maximum Weight} scheduling policy (MSMW for short).
\subsection{Time Complexity of Throughput-optimal Solutions}

The design of throughput-optimal policy is NP-complete.
We consider the case where $\delta_i\rightarrow\infty$ for all links.
Then our problem is same with the basic scheduling problem in general wireless networks,
which is proved to be NP-complete \cite{gupta2010delay}.
To achieve an optimal solution means the algorithm must find a max weight independent set in a general graph in each slot,
which can only be solved in exponential time.

\subsection{MSMW Scheduling Policy}
MSMW needs to take care of service frequency and queue length of each link.
Therefore, it has two parameters: \textit{stage} and \textit{weight}.

\textit{stage} shows whether a link has been scheduled in the current frame and the number of time slots to the end of the frame.
The stage of link $i$ in slot $t$ is calculated as follows:

\begin{equation}
ST_i[t]=
\begin{cases}
0 &\text{when link $i$ is scheduled}\\
&\text{\ in the current frame}\\
\ \\
\delta_i - (t-1)\ mod\ \delta_i &\text{others}\\
\end{cases}
\end{equation}
where $\delta_i$ is the service frequency constraint of link $i$.
In each slot, MSMW first considers the links with small stages which have not been scheduled.
Then MSMW can continue in a round-robin manner since unscheduled links will eventually obtain a higher priority.

\textit{weight} represents the traffic load of each link, as is used by most of the previous works.
MSMW serves the link with the largest weight when multiple links are in the same stage.
We define \textit{weight} of link $i$ at time slot $t$ as follows:
\begin{equation}
w_i[t]={Q_i[t]}
\end{equation}
We will use $w_i[t]$ and $Q_i[t]$ alternately in the rest of the paper.

In each time slot, MSMW has two phases:

\textit{Phase 1 : Update}

At the beginning of each time slot, MSMW updates \textit{stage} and \textit{weight} for each link with Equations (4) and (5), 
based on the transmission information in the last slot and the arriving packets in the current slot.
Then the policy partitions the links into different groups according to their stages.
We denote the groups as $GP=\{gp_0,gp_1,gp_2,\cdots,gp_M\}$. 
The links in a same group have identical stage.
We also assume that the stage of links in each group increases with the group number, \emph{i.e.}, $gp_1$ includes the links with the minimum stage while $gp_M$ contains the links with the largest stage.
$gp_0$ is composed of the links whose stages are zero.
All the links will fall in $gp_0$ when they have been scheduled in the current frame.

\textit{Phase 2: Scheduling}

MSMW iteratively selects the links from group $gp_M$ when $M=0$, or from $gp_1$ to $gp_M$ and then $gp_0$ when $M\neq 0$.
During each iteration, MSMW chooses a link from first group containing links with following equation,
\begin{equation}
l_i=argmax_{l_j\in gp_k}w_j[t],
\end{equation}
where $gp_k$ is the first group still has links remaining.
After the selection, MSMW deletes all the conflicted links of $l_i$ and starts next iteration.
The scheduling phase ends when there is no link left in current slot.
\ \ \ $\blacksquare$\\

Now we use a simple example to demonstrate the scheduling policy.
Suppose there are $2$ links in the network, with parameters $A_1[t]=\lambda_1$, $A_2[t]=\lambda_2$, $\frac{\lambda_1}{r_1c_1}=\frac{1}{2}$, and $\frac{\lambda_2}{r_2c_2}=\frac{1}{8}$. $\delta_1 = 2$ and $\delta_2 = 4$. 
The process is shown in Fig.1.
In slot $1$, $ST_1[1] = 2$ and $ST_2[1]=4$.
Then $l_1$ is scheduled since it has a smallest stage.
In slot $2$, $ST_1[2] = 0$ and $ST_2[2]=3$.
Then $l_2$ is chosen.
In slot $3$, $ST_1[3] = 2$ since a new frame of $l_1$ begins.
$ST_2[3]=0$.
Then $l_1$ is selected.
In slot $4$, $ST_1[4] = 0$ and $ST_2[4]=0$. $w_1[4]= 1/2$ and $w_2[4]= 1/4$.
Then $l_1$ is scheduled since it has a larger weight.
The following process is similar, thus omitted here.

\begin{figure}
\centering
\includegraphics[scale=0.4]{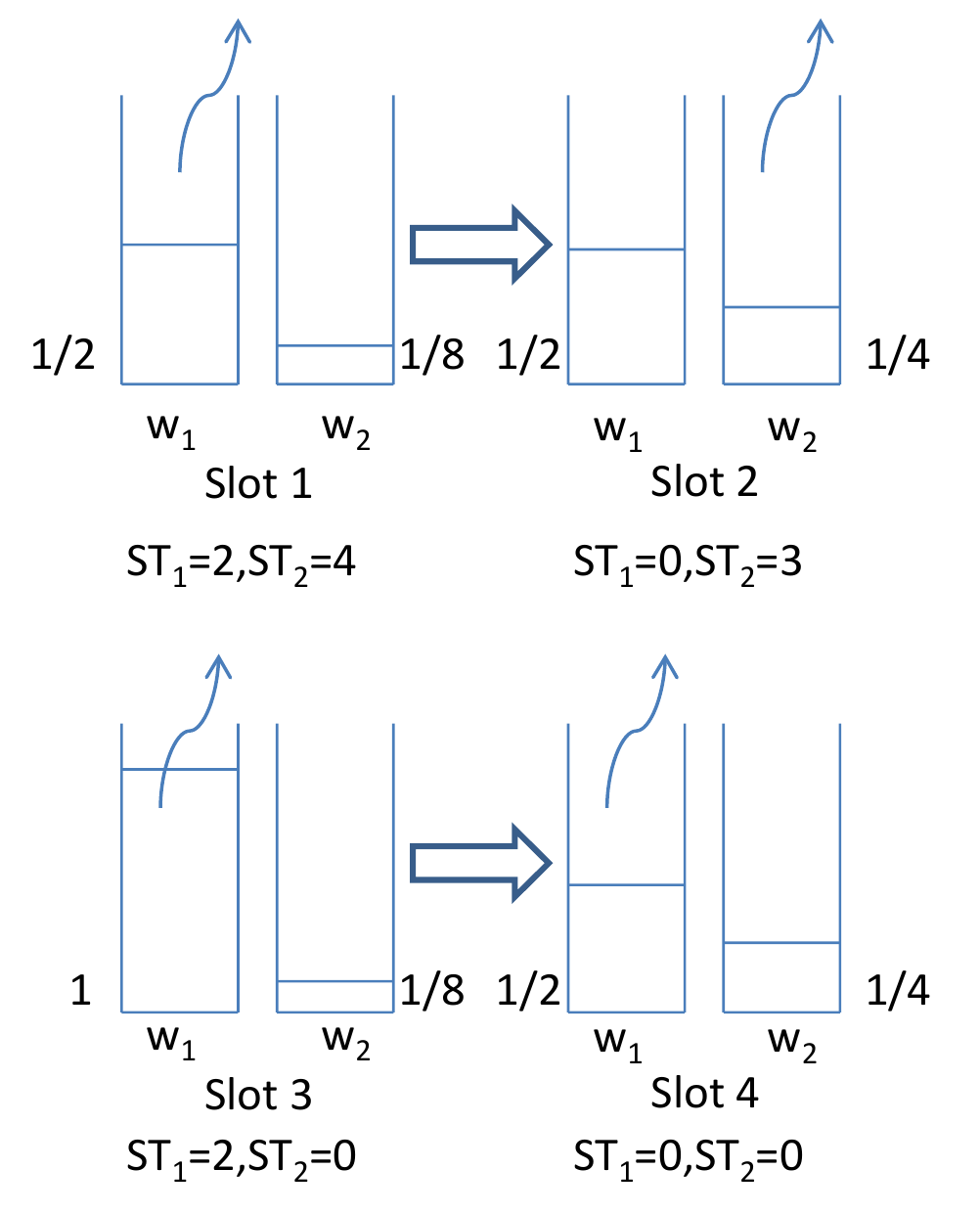}
\caption{An Instance of MSMW.}
\label{fig:example}
\end{figure}

The computational complexity of MSMW is analyzed as follows.
The time complexity of updating all the stages and weights is $O(N)$.
Partitioning links into groups also costs $O(N)$ scale of time.
During phase $2$, the time complexity is $O(N^2)$ in the worst case, and $O(1)$ when there are constant number of links in the first group.
As a whole, the time complexity of MSMW in each time slot is $O(N^2)$.

\subsection{Performance of MSMW in Collocated Networks}
In this part, we demonstrate that MSMW is throughput optimal in a specific type of networks, \textit{i.e.}, collocated networks \cite{13interservicetimeli2013throughput}.
In a collocated network, all links conflict with each other, and only one link can be scheduled in each slot.
This model is used by systems like cellular networks, which is an fundamental architecture to support the environment monitoring systems.
We first show in Lemma 2 that all the supportable rate vectors follow the second condition in \textit{Definition 1} when a collocated network employs the MSMW scheduling policy.
Then in Lemma 5, all these vectors are proved to follow the first condition, which means the underlying queue length is finite.

We proposes a necessary condition for all the arriving rate vectors in the capacity region.
This condition means that for a sufficient long time, the number of slots used to meet the service frequency constraints of all the links is no more than the total number of time slots.

\begin{lemma}
The service frequency constraints of all the links must follow equation
\begin{equation}
\sum_{i=1}^{N}\frac{1}{\delta_i}\leq 1,
\end{equation}
if the rate vector is supportable by the network.
\end{lemma}
\begin{proof}
The proof is by contradiction.
Assume that $\sum_{i=1}^{N}\frac{1}{\delta_i}> 1$.
Then the total services that all the links get are no less than $\sum_{i=1}^{N}\frac{T}{\delta_i}$ in $T$ time slots.
This is due to the fact that each link $i$ gets a service every $\delta_i$ slots.
However, it is infeasible since $\sum_{i=1}^{N}\frac{T}{\delta_i}=T\sum_{i=1}^{N}\frac{1}{\delta_i}>T$.
It contradicts the fact that the arriving rate vector is supportable.
\end{proof}

Lemma 2 shows that MSMW can ensure the constraints when the arriving rate vector follows Lemma 1.

\begin{lemma}
The service frequency constraints of the links in a network is achievable under MSMW if $\sum_{i=1}^{N}\frac{1}{\delta_i}\leq 1$.
\end{lemma}
\begin{proof}
Assume link $l_p$ does not get service in the $k$th interval that starts at slot $t_k$.

Consider an arbitrary link $l_i$.
It will not receive service in the interval that includes the final slot of $k$th interval of $l_p$.
The reason is that $l_p$ always has a smaller stage than $l_i$ in that interval of link $l_i$.
It will only be scheduled once in a interval if the time period is totally covered by the $k$th interval of $l_p$.
If $l_i$ has a $j$th interval that covers slot $t_k$ and its scheduling in this interval happens no early than slot $t_k$, all the other links can only be scheduled once per interval during the slots of $j$th interval for $l_i$.

We can iteratively look for such intervals back and find the first link $l_x$ and its $y$th interval follow either of the following principals: 1) all the other links have new intervals that start at the same slot with  $y$th interval of link $l_x$, 2) all the intervals that cover $t_y$ are scheduled early than slot $t_y$.

Now we consider the time between $t_y$ and $t_k+\delta_p$.
Every link can only be scheduled once in each interval that is totally covered by this time period.
This result can be combined with the fact that link $l_p$ does not get service in the $k$th interval:  
\begin{equation*}
\sum_{i=1}^{N}\lfloor{\frac{t_k+\delta_p}{\delta_i}}\rfloor -1\leq \sum_{i=1}^{N}{\frac{t_k+\delta_p}{\delta_i}} -1\leq t_k+\delta_p-1<t_k+\delta_p
\end{equation*}
Then the number of feasible services is less than the size of time slots, which is a contradiction.
It indicates a link could be scheduled at least once per interval when $\sum_{i=1}^{N}\frac{1}{\delta_i}\leq 1$.
\end{proof}

MSMW can satisfy the second condition in \textit{Definition 1} for all the rate vectors inside the capacity region since it is a subset of the vectors following Lemma 1.

Now we analyze the performance of MSMW on the first condition in \textit{Definition 1}.
Our proof is based on the conclusion in \textit{Theorem 1}.
MSMW actually runs in two manners: the maximum weight manner when it chooses a link from $gp_0$, and the round-robin manner otherwise.
As we know, the maximum weight scheduling is throughput-optimal since it always chooses the link with the globally maximum weight.
Our work turns to prove that the queue length will never be infinite when we scheduling a link not in $gp_0$.
Specially, we need to discuss the performance of the system when the queue length of scheduled link is less than its transmission rate.

We start with another necessary condition for any supportable rate vector.
It bounds the size of a packet that cannot be delivered if every link only transmits once in every frame.
Notice that Lemma 3 is not restricted to MSMW.

\begin{lemma}
For any supportable rate vector, it must follow that
$\sum_{1}^{N}Max\{\frac{\lambda_i}{r_ic_i}-\frac{1}{\delta_i},0\}< 1-\sum_{1}^{N}\frac{1}{\delta_i}.$
\end{lemma}
\begin{proof}
The proof is straightforward.
For a sufficient long time $T$, the expectation of the size of a packet that fails to be delivered equals to $\sum_{1}^{N} Max\{T\times\lambda_i-\frac{T\times r_ic_i}{\delta_i},0\}$ 
if each link is only scheduled once in every frame.
Delivering these packets totally costs $T\times\sum_{1}^{N}Max\{\frac{\lambda_i}{r_ic_i}-\frac{1}{\delta_i},0\}$ slots.
Meanwhile, the number of slots left is $T-\sum_{1}^{N}\frac{T}{\delta_i}$.
Thus, $\sum_{1}^{N}Max\{\frac{\lambda_i}{r_ic_i}-\frac{1}{\delta_i},0\} < 1-\sum_{1}^{N}\frac{1}{\delta_i}$ must hold.
\end{proof}

In Lemma 4, we show that the expectation of the total queue length is bounded in slot $t$ when a link $i$ with $ST_i[t] \neq 0$ is scheduled.
Denote $T_0$ as the least common multiple of $\delta_1, \delta_2, \cdots, \delta_N$.
\begin{lemma}
When MSMW schedules a link not in $gp_0$ in slot $t$, the expectation of the total queue length $E(\sum_{i=1}^{N}Q_i[t])\leq B$, where $B=N+h*B_1+\sum_{i=1}^{N}(\sum_{j=1}^{N}\frac{1}{\delta_j}*\frac{\lambda_i}{r_ic_i}*T_0)$. $B$ and $h$ are constants, and $B_1=\frac{1}{2(1-\sum_{j=1}^{N}\frac{1}{\delta_j})^2}\sum_{l_i\in\{l_j|\frac{\lambda_j}{r_jc_j}-\frac{1}{\delta_j}>0\}}[(1-\sum_{j=1}^{N}\frac{1}{\delta_j})(\frac{\lambda_i}{r_ic_i}-\frac{1}{\delta_i})+Var(\frac{A_i[t]}{r_ic_i}-\frac{1}{\delta_i})-(\frac{\lambda_i}{r_ic_i}-\frac{1}{\delta_i})^2]$.
\end{lemma}
\begin{proof}
The proof of lemma is composed of two parts.
The first part estimates the total queue length of the system at slot $k*T_0$, where $k$ is a positive integer.
The second part gives the upper bound of $E(\sum_{i=1}^{N}Q_i[t])$ when MSMW schedules a link not in $gp_0$.
The bound also holds for the time when MSMW selects a link in $gp_0$.
We discard this case since it has nothing to do with our subsequent proofs.
The details of the proof is given in Appendix A.
\end{proof}

According to Lemma 4, the total queue length of the system is bounded by a constant $B$ when MSMW schedules a link in $gp_i,i\neq 0$.
Now we can prove that for any supportable rate vector, the underlying Markov chain is positive recurrent under MSMW.
The proof is based on the Foster's criteria in Theorem $1$.

\begin{lemma}
For any supportable arriving rate vector, the underlying Markov chain $Q[t]=\sum_{i=1}^{N}Q_i[t]$ is positive recurrent when we apply MSMW to schedule the links.
\end{lemma}
\begin{proof}
We use the Lyapunov function defined as
\begin{center}
	$V(Q[t])=\frac{1}{2}\sum_{i=1}^{N}Q_i[t]^2.$
\end{center}

Then
\begin{equation}
\begin{split}
\Delta(Q[t])&=E(V(Q[t+1])-V(Q[t])|\boldsymbol{Q[t]})\\
&=\frac{1}{2}\sum_{i=1}^{N}E((Q_i[t+1]-Q_i[t])\\
&\text{\ \ \ \ \ \ \ \ \ \ \ \ \ \ }*(Q_i[t+1]+Q_i[t])|\boldsymbol{Q[t]})\\
&=\frac{1}{2}\sum_{i=1}^{N}E((A_i[t+1]-r_ic_is[t])\\
&\text{\ \ \ \ \ \ \ \ \ \ \ \ \ \ }*(2Q_i[t]+A_i[t+1]-r_ic_is[t])|\boldsymbol{Q[t]})\\
&=\frac{1}{2}\sum_{i=1}^{N}E((A_i[t+1]-r_ic_is[t])\\
&\text{\ \ \ \ \ \ \ \ \ \ \ \ \ \ }*(2Q_i[t]+A_i[t+1]-r_ic_is[t])|\boldsymbol{Q[t]})\\
&=\sum_{i=1}^{N}E(Q_i[t](A_i[t+1]-r_ic_is[t])|\boldsymbol{Q[t]})\\
&\text{\ \ \ \ \ \ \ \ \ \ \ \ \ \ }+\frac{1}{2}\sum_{i=1}^{N}E((A_i[t+1]-r_ic_is[t])^2|\boldsymbol{Q[t]})
\end{split}
\end{equation}

The second part of the last equation in (8) is bounded since the second moment of $A_i[t]$ is finite and it is independent of the current queue size $Q[t]$.
The first part is also bounded by $B*\sum_{i=1}^{N}E(A_i[t+1]-r_ic_is[t])$ if the links not in $gp_0$ are scheduled.
Denote $\epsilon_{rr}$ as the set of  all $Q[t]$s that MSMW schedules a link not in $gp_0$.
Then when $Q[t]\in\epsilon_{rr}$, the drift $\Delta(Q[t])$ is bounded.

When $Q[t]\notin\epsilon_{rr}$, MSMW schedules the link with the globally maximum weight.
Then the system is stable according to the conclusion in \cite{gupta2010delay}.
As proved by the previous work, there is a set of states $\epsilon_{mw}$ in which the drift $\Delta(Q[t])$ is bounded.
And for $Q[t]\notin\epsilon_{mw}$, $\Delta(Q[t])<0$.

Finally, we have

(1) when $Q[t]\in\epsilon_{rr}\cup \epsilon_{mw} $, $\Delta(Q[t])<\infty$,

(2) when $Q[t]\notin\epsilon_{rr}\cup \epsilon_{mw} $, $\Delta(Q[t])<0$.
\end{proof}

MSMW can support all the rate vectors in the capacity region by combining Lemma 2 and Lemma 5.
We conclude it in the following theorem, and the proof is straightforward.

\begin{theorem}
MSMW is a throughput-optimal scheduling policy under the proposed network model and the definition of stability.
\end{theorem}

\section{Discussion}
In this section, we continue our analysis on the performance of a collocated network with the service frequency constraints.
Our objectives are mainly on three aspects.
In part A, we consider the change of the network performance from an overall view.
We discuss the change on the set of supportable arriving rate vectors when introducing a group of fixed service frequency constraints.
We compare the size of the new capacity region with the original region without the constraints.
In part B, we aim to check the in-depth performance of the system with service frequency constraints.
We check the total queue length of the system under a rate vector that is supportable both with or without the constraints.


\subsection{Decreasing in Capacity Region}
The constraint on service frequency reduces the size of capacity region.
This is due to the fact that some links with few packets can still be scheduled to meet the constraints, and then the bandwidth is wasted.
Other links have to wait and the length of their queues could possibly be infinite.
A larger $\sum_{i=1}^{N}\frac{1}{\delta_i}$ often leads to larger decrease in capacity region since more slots will be used by the links not in $gp_0$ to keep the service frequency.
We formulate the decrease in Theorem $4$.
Lemma $6$ estimates the size of the vector space composed of all vectors $k=(k_1,k_2,\cdots,k_N)$ that follow $k_i\geq 0,\forall i$ and $\sum_{1}^{N}k_i\leq\beta$.
The conclusion in Lemma $6$ is applied for the proof of Theorem 4.

\begin{lemma}
For any vector $k =(k_1,k_2,\cdots,k_n)$, $\oint_{0}^{\beta}1d\sigma=\frac{\beta^n}{n!}$ if $k_i\geq 0,\forall i$ and $\sum_{1}^{N}k_i\leq\beta$. $\oint_{0}^{\beta}1d\sigma$ is a multiple integration on the $n$ dimensions, also defined as the size of the vector space composed of all vectors like $k$.
\end{lemma}
\begin{proof}
We prove the lemma by induction.

When $n=1$, $\oint_{0}^{\beta}1d\sigma=\beta$. When $n = 2$, $\oint_{0}^{\beta}1d\sigma=\frac{\beta^2}{2}$.

Assume  $\oint_{0}^{\beta}1d\sigma=\frac{\beta^k}{k!}$ is true when $n=k$.
When $n = k+1$,
\begin{equation}\begin{split}
 \oint_{0}^{\beta}1d\sigma & =\int_{0}^{\beta}\oint_{\beta}^{0}1d\sigma'dx
\end{split}\end{equation}
where $\sigma'$ is a $k$-dimensional space with $\sum_{i=1}^{k}(r_ic_i)\leq\beta-r_{k+1}c_{k+1}$, and
\begin{equation}\begin{split}
 \int_{0}^{\beta}\oint_{\beta}^{0}1d\sigma'dx & =\int_{0}^{\beta}\frac{(\beta-x)^n}{n!}dx=-\int_{\beta}^{0}\frac{(\beta-x)^n}{n!}d(\beta-x)\\
 &=-\int_{\beta}^{0}\frac{(y)^n}{n!}d(y)=-\frac{1}{(n+1)!}(y)^{n+1}|_\beta^0\\
 &=\frac{\beta^{n+1}}{(n+1)!}
\end{split}\end{equation}
\end{proof}

Theorem $4$ gives the ratio between the new capacity region and the original one without service frequency constraints.
\begin{theorem}
The size of capacity region $\Lambda$ of a network is no more than $N!\sum_{t=0}^{N}C_N^t[\frac{\epsilon^t}{t!}(\frac{1}{\delta_{min}})^{N-t}]$ of its original capacity region $\Lambda_0$, where $\epsilon = 1-\sum_{1}^{N}\frac{1}{\delta_i}$, and $\delta_{min}$ is the minimum service frequency constraint among all the links.
\end{theorem}

\begin{proof}
Based on Lemma $3$, $\sum_{i=1}^{N}Max\{\frac{\lambda_i}{r_ic_i}-\frac{1}{\delta_i},0\}< 1-\sum_{1}^{N}\frac{1}{\delta_i}$ holds for any rate vector inside the capacity region. We denote $\frac{\lambda_i}{r_ic_i}$ as $k_i$, and rewrite the inequality as follows:
\begin{equation}
\sum_{i=1}^{N}Max\{k_i-\frac{1}{\delta_i},0\}<\epsilon
\end{equation}

We will first discuss the size of the vector spaces when a fixed proportion of links have $k_i$s larger then each $\delta_i$.
Then we can achieve the size of the capacity region by adding up all the vector spaces.	
Assume there are $t$ users whose $k_i \geq\frac{1}{\delta_i}$. $t$ changes from 0 to $N$, and we define $\sigma_t$ as the size of the corresponding vector space.
\begin{equation}
\begin{split}
\sigma_t&=\sum_{i=1}^{|C_t|}[\oint_{0}^{\epsilon}1d\sigma_t'\prod_{j\notin C_{ti}}(\frac{1}{\delta_j})]\\
\end{split}
\end{equation}

In (11), $C_t=\{C_{t1},C_{t2},\cdots,C_{t|C_t|}\}$, where $C_{ti}$ is the set of arbitrary $t$ links in a network.
There are in total $|C_t|$ different sets.
It is easy to see that $|C_t|=C_N^t$,
where $C_N^t$ is the combinatorial number of picking $t$ out of $N$.
$\oint_{0}^{\epsilon}1d\sigma_t'$ is the size of the space composed by the vector $\{k_1-\frac{1}{\delta_1},k_2-\frac{1}{\delta_2},\cdots,k_t-\frac{1}{\delta_t}\}$, where $\sum_{i=0}^{t}(k_i-\frac{1}{\delta_i})\leq\epsilon$, and $k_i\geq\frac{1}{\delta_i}$ for $i =1,\cdots,t$.
It means there are $t$ links whose $k_i\geq\delta_i$.
Finally, $\prod_{j\notin C_{ti}}(\frac{1}{\delta_j})$ is the vector space of the users whose $k_i<\frac{1}{\delta_i}$.

According to Lemma 4, $\oint_{0}^{\epsilon}1d\sigma_t'=\frac{\epsilon^t}{t!}$. Add up all the $\sigma_t$ from $t=0$ to $t = N$,
\begin{equation}
\begin{split}
\sigma&=\sum_{t=0}^{N}\sigma_t=\sum_{t=0}^{N}\sum_{i=1}^{|C_t|}[\oint_{0}^{\epsilon}1d\sigma_t'\prod_{j\notin C_{ti}}(\frac{1}{\delta_j})]\\
&=\sum_{t=0}^{N}\sum_{i=1}^{|C_t|}[\frac{\epsilon^t}{t!}\prod_{j\notin C_{ti}}(\frac{1}{\delta_j})]\leq\sum_{t=0}^{N}\sum_{i=1}^{|C_t|}[\frac{\epsilon^t}{t!}(\frac{1}{\delta_{min}})^{N-t}]\\
&\leq \sum_{t=0}^{N}C_N^t[\frac{\epsilon^t}{t!}(\frac{1}{\delta_{min}})^{N-t}].
\end{split}
\end{equation}

The size of the original capacity region is $\Lambda_0=\frac{1}{N!}$, which is derived by applying $\sum_{i=1}^{N}k_i\leq 1$ to Lemma 4.
Thus, $\Lambda/\Lambda_0\leq N!\sum_{t=0}^{N}C_N^t[\frac{\epsilon^t}{t!}(\frac{1}{\delta_{min}})^{N-t}]$.
\end{proof}

Theorem 4 shows that the scale of the capacity region keeps increasing as $\epsilon$ becomes larger.
When $\epsilon\rightarrow 1$, we have $\sum_{1}^{N}\frac{1}{\delta_i}\rightarrow 0$, then $\sum_{t=0}^{N}\sum_{i=1}^{|C_t|}[\frac{\epsilon^t}{t!}\prod_{j\notin C_{ti}}(\frac{1}{\delta_j})]\rightarrow \frac{1^N}{N!}$.
It means that $\Lambda/\Lambda_0=1$ when the service frequency constraints are sufficiently loose for all the links.
On the other hand, when $\epsilon\rightarrow 0$, we have $\sum_{1}^{N}\frac{1}{\delta_i}\rightarrow 1$,
then $\sum_{t=0}^{N}\sum_{i=1}^{|C_t|}[\frac{\epsilon^t}{t!}\prod_{j\notin C_{ti}}(\frac{1}{\delta_j})]=\prod_{i=1}^{N}(\frac{1}{\delta_i})$.
$\prod_{i=1}^{N}(\frac{1}{\delta_i})$ is maximized when all the $\delta_i$ are identical.
It indicates that with tough service frequency requirements, the capacity region is maximized when all the links have same constraints $\delta = N$, which refers to a strict round-robin scheduling.

\subsection{Performance on Queue Length}
The queue length of a system under an arriving rate vector usually increases with the existence of service frequency constraints.
This is aroused by the scheduling of the links with very small queue length.
These links cannot take full use of the channel, which will lead to the increasing of the global queue length.
We estimate an upper bound on the size of the total queue length of the links under fixed $\{\lambda_1,\lambda_2,\cdots,\lambda_N\}$ and $\{\delta_1,\delta_2,\cdots,\delta_N\}$.

\begin{theorem}
For an arriving rate vector $\lambda=\{\lambda_1, \lambda_2,\cdots, \lambda_N\}$, and service frequency constraints $\{\delta_1,\delta_2,\cdots,\delta_N\}$, the expectation of the average total queue length in a long time follows $\lim_{t\rightarrow \infty}\frac{\sum_{t'=1}^{t}\sum_{i=1}^{N}Q_i[t']}{t}\leq B_3$ if the system employs MSMW.
Assume $\epsilon=1-\sum_{i=1}^{N}\frac{1}{\delta_i}$.
Then $B_3 = \frac{N+1}{2}+\frac{N'}{8}+\frac{1}{2\epsilon^2}\sum_{l_i\in L'}\frac{Var(A_i[t])}{(r_ic_i)^2}$ when $\epsilon>0$, where $L'$ is the set of links with $\frac{\lambda_j}{r_jc_j}-\frac{1}{\delta_j}>0$, $N'$ is the size of $L'$.
$B_3=\sum_{i=1}^{N}(Min\{\frac{\lambda_i}{r_ic_i},\frac{1}{\delta_i}\}*\frac{(1+\delta_i)}{2})$ when $\epsilon=0$.
\end{theorem}
\begin{proof}
See Appendix B.
\end{proof}

As we see, the queue length is affected by many factors, like the number of the total links, the sum of $\frac{1}{\delta_i}$'s and so on.
To derive the accurate influence of these factors will be one of our future works.


\section{Evaluation}
In this section, we evaluate MSMW through a group of simulations.
We mainly compare MSMW with a state-of-the-art scheduling policy \textit{Regulated Throughput Optimal} (\textit{RTO}) \cite{13interservicetimeli2013throughput}.
RTO is also a combination of the two fundamental scheduling policies: max weight scheduling and round-robin scheduling.
RTO serves the link with the maximum weight in each slot.
The weight of each link is a combination of queue length $Q_i[t]$ and the time-since-last-service $T_i[t]$.
The time-since-last-service is the number of time slots that start from the link's last transmitting.
There are two weighting factors, $\alpha_i$ and $\beta_i$, for the two components in the weight.
We set $\alpha_i = \frac{1}{r_ic_i}$ and $\beta_i = \frac{1}{\delta_i}$ for each link, \emph{i.e.}, $w_i[t]=\frac{1}{r_ic_i}Q_i[t]+\frac{1}{\delta_i}T_i[t]$.
Since RTO only applies to collocated networks, we detailedly compare MSMW with RTO in this type of network.
Furthermore, we also give the performance of MSMW in general networks.
In every group of simulations, we run each scheduling policy ten times, with $10000$ slots each time.

Our evaluation of MSMW mainly focuses on two metrics:
the performance on service frequency constraints and the total queue length of the system.
As proved, MSMW can keep the service frequency constraints for all the links.
We are also interested in the performance of RTO on the service frequency, as it can regulate the services among links.
Meanwhile, we check the queue length since it can tell the stability of the system.
MSMW ensures service frequency at the cost of increasing total queue length.
Such a cost is tolerable in our simulation results.

\subsection{Number of Links}
The first group of simulations test MSMW and RTO with different numbers of links, which are $4$, $8$, $16$, $32$, and $64$.
The service frequency constraint of each link is $N+1$, where $N$ is the number of the links.
We set $r_i = 2$ and $c_i=0.5$ for link $i$ ($i=1,2,\cdots,N$)
and $E(A_i)=\frac{1}{N*i}$ for the $i$th link.
The arriving rates are asymptotically decreasing.
The result is shown in Fig.2.

As can be seen, MSMW can satisfy the service frequency constraints for all the links, while the ratio of the links that satisfy the constraints in RTO decreases as $N$ increases.
It falls from $1$ ($4$ out of $4$ links) to $0.109375$ (only $7$ out of $64$ links).
It means that more and more links have slowly increasing weights in RTO due to their small arriving rates.
Their weights are not large enough to always be selected by the RTO within the service frequency constraints.
The performance of MSMW is not affected since it is independent of the arriving rates.

At the same time, we find that the total queue length and maximum queue length of MSMW are comparable with those of RTO.
It is only a bit larger (less than $0.2$ for the total queue length and $0.3$ for the maximum queue length).
The increasing is mainly caused by the scheduling towards the links with very small queue length, which potentially wastes bandwidth.
The queue lengths are equal when $N=4$, which means the two policies have identical performance.

\begin{figure}[htbp]
\centering
\subfloat[Ratio of Links Achieving Service Frequency]{
\label{fig:improved_subfig_a}

\centering
\includegraphics[scale=0.16]{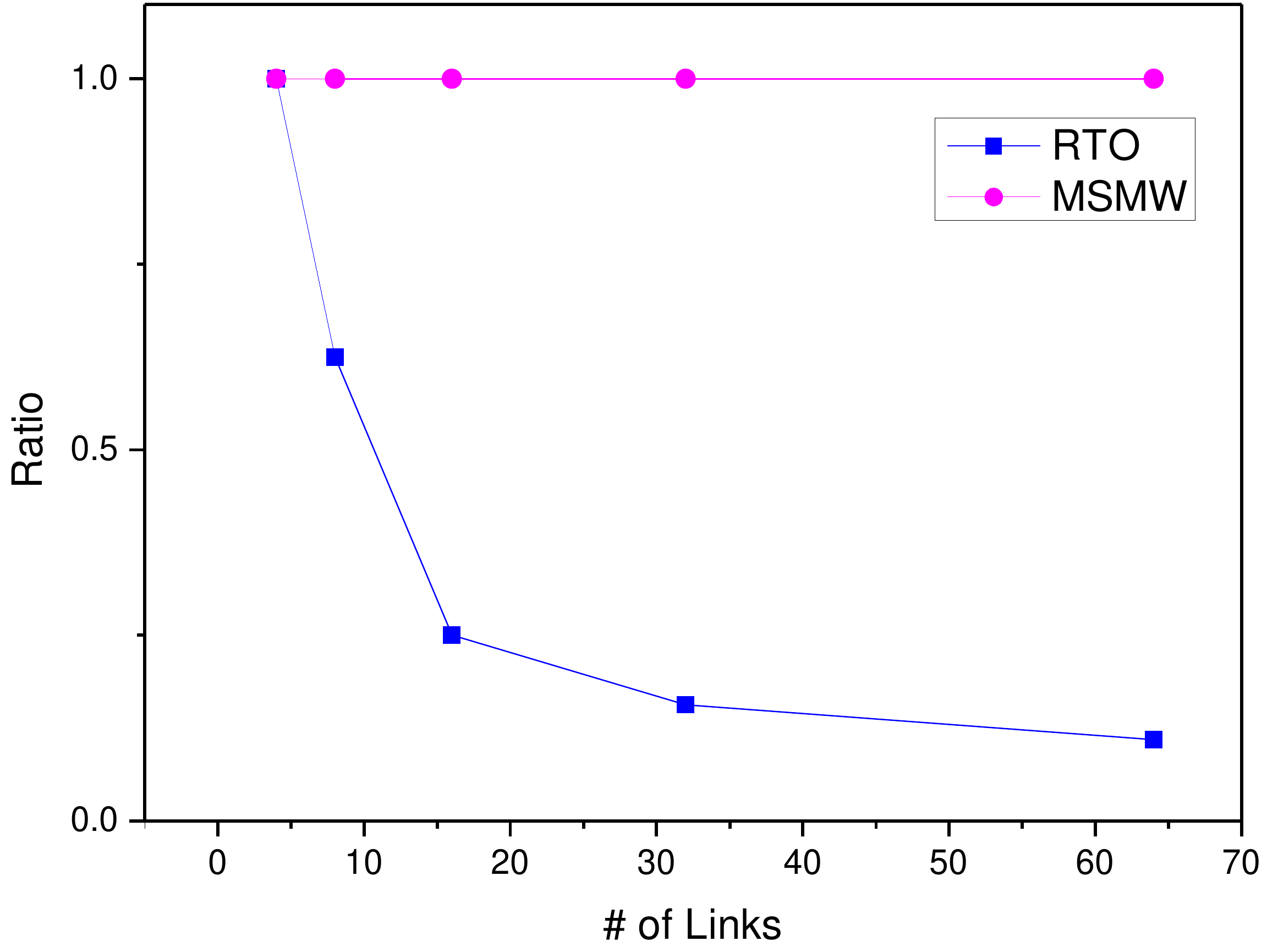}

}
\subfloat[Queue Length]{
\label{fig:improved_subfig_b}
\centering
\includegraphics[scale=0.16]{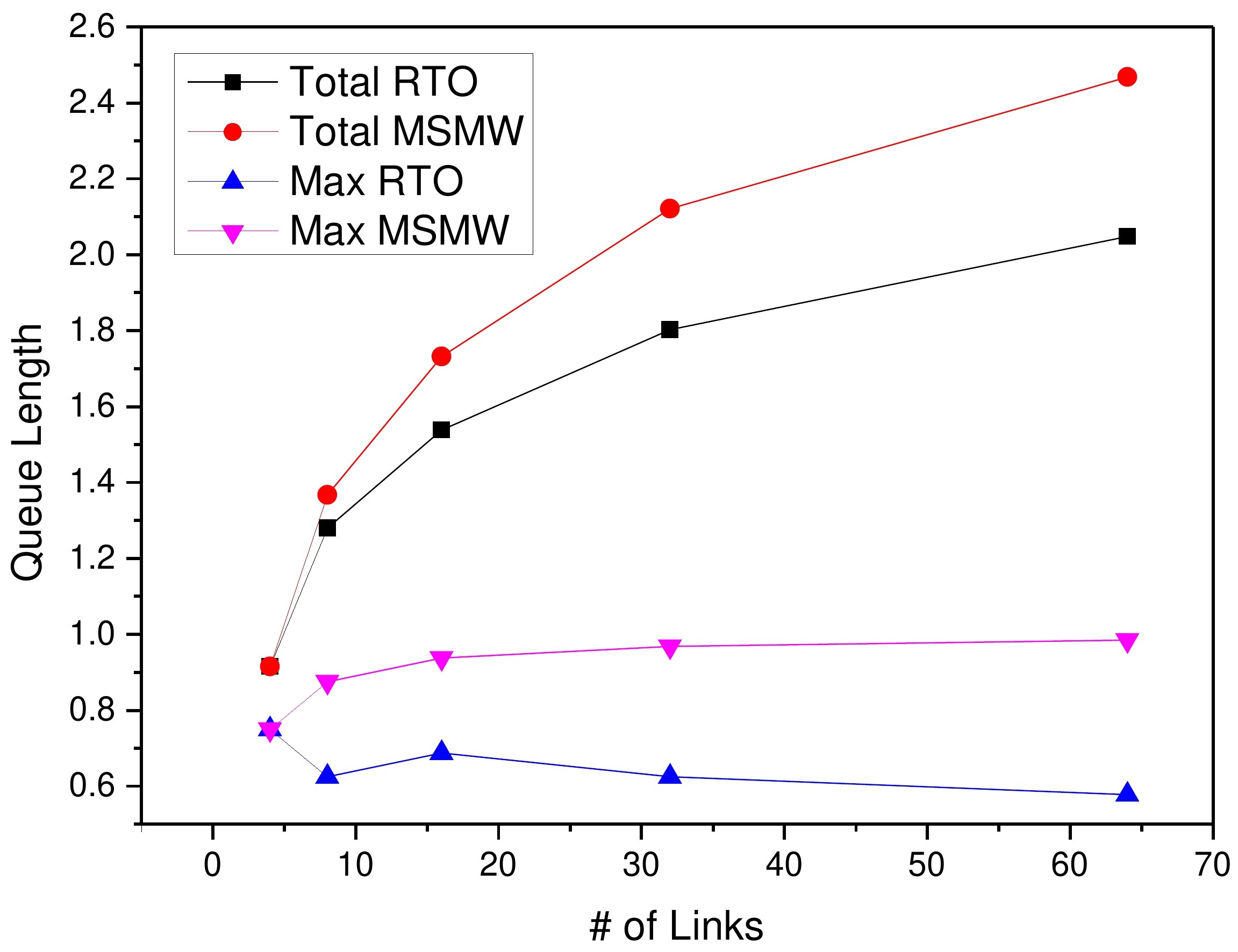}

}
\caption{System performance with different numbers of links}
\end{figure}

\subsection{Service Frequency Constraints}
The second group of simulations evaluates the performance of MSMW with different service frequency constraints.
There are $32$ links in the network, with $E(A_i)=\frac{1}{N*i}$ for link $i$.
We set the service frequency constraint to be $N+k$, and $k = 1,2,3,4,5$.
The result is shown in Fig.3.

We find that the performance of RTO rapidly improves as $k$ increases, while MSMW always satisfies all the constraints.
The reason for the improvement in RTO is that the selected $\alpha_i$'s and $\beta_i$'s happen to be more fitful for larger $k$'s.
It is true that RTO can keep the service frequency if we carefully set the parameters according to the fixed arriving rates and service frequency constraints.
The queue length of MSMW is still comparable with that of RTO, as shown in Fig.3.

\begin{figure}[htbp]
\centering
\subfloat[Ratio of Links Achieving Service Frequency]{
\label{fig:improved_subfig_a}

\centering
\includegraphics[scale=0.16]{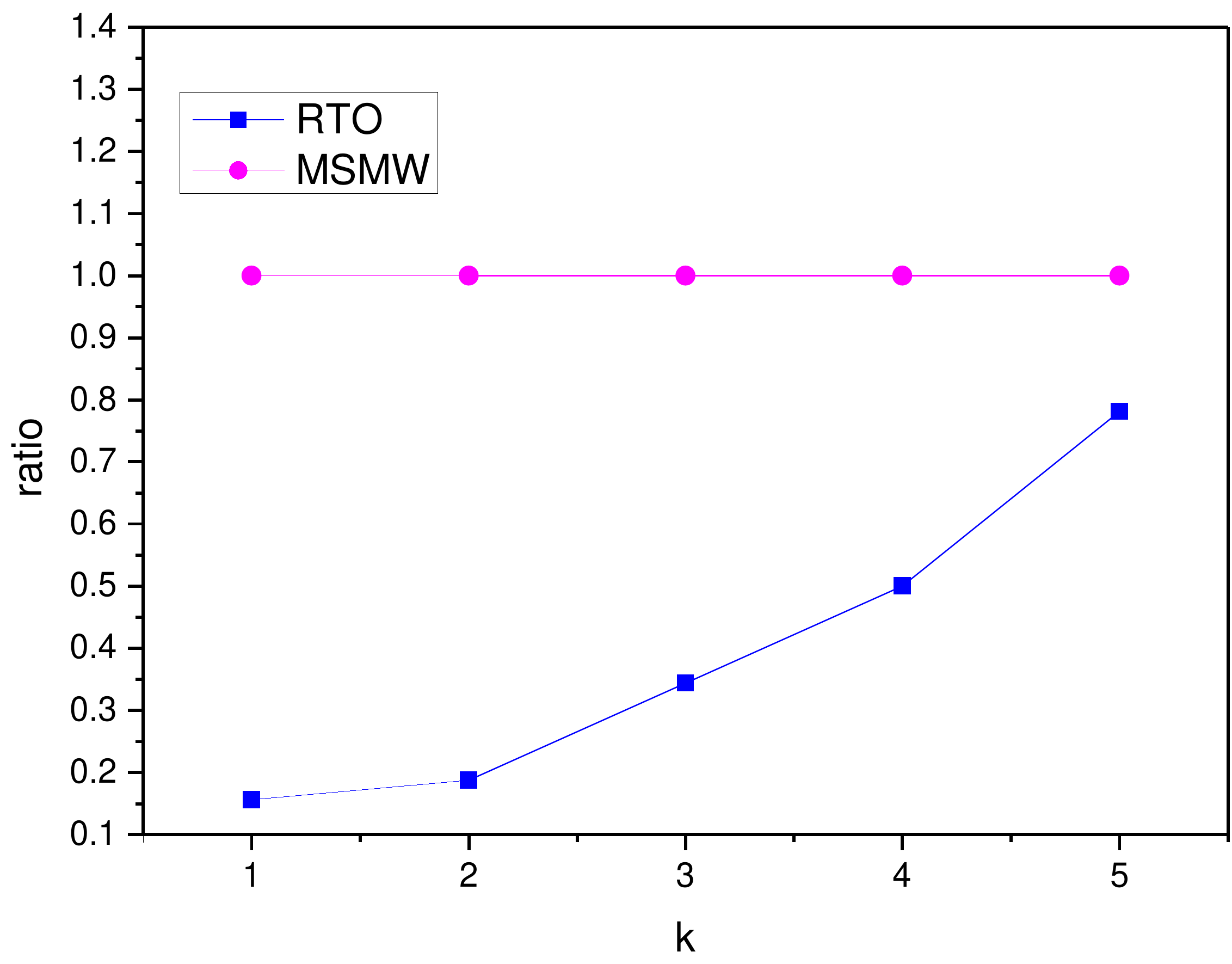}

}
\subfloat[Queue Length]{
\label{fig:improved_subfig_b}
\centering
\includegraphics[scale=0.16]{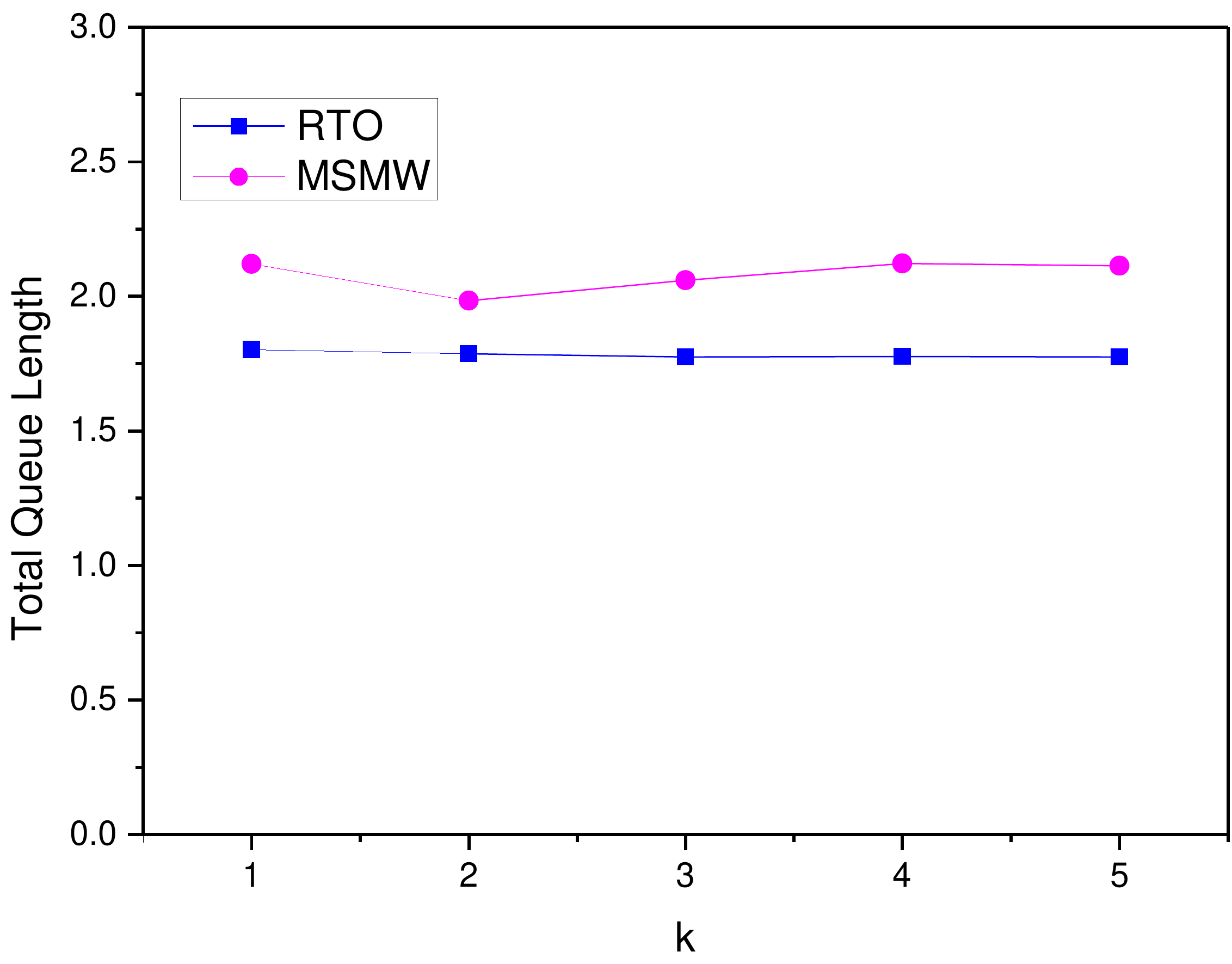}

}
\caption{System performance with various service frequency constraints}
\end{figure}

\subsection{Arriving Rates}
In the third group of simulations, we evaluate the performance with different arriving rate vectors.
There are in total $16$ links in the network, and the service frequency constraint is $N+5$ for all the links.
We set the arriving rates of the first 8 links to be $E(A_i)=\frac{1}{N*i}$, and the last 8 links to be $E(A_i)=\frac{1}{N*(i+k)}$ ($k = 1,2,\cdots,20$).
We call this rate vector two tails since the two parts of the links have different scales of arriving rates.
The difference between the two parts becomes larger as $k$ increases.
We show the results in Fig.4.

MSMW can still satisfy the constraints in all the situations, while RTO's performance degrades as $k$ increases.
RTO can only ensure service frequency for the first part when $k \ge 17$.
The queue length of MSMW is larger than that of RTO by a small fraction (around 10\%).

\begin{figure}[htbp]
\centering
\subfloat[Ratio of Links Achieving Service Frequency]{
\label{fig:improved_subfig_a}

\centering
\includegraphics[scale=0.16]{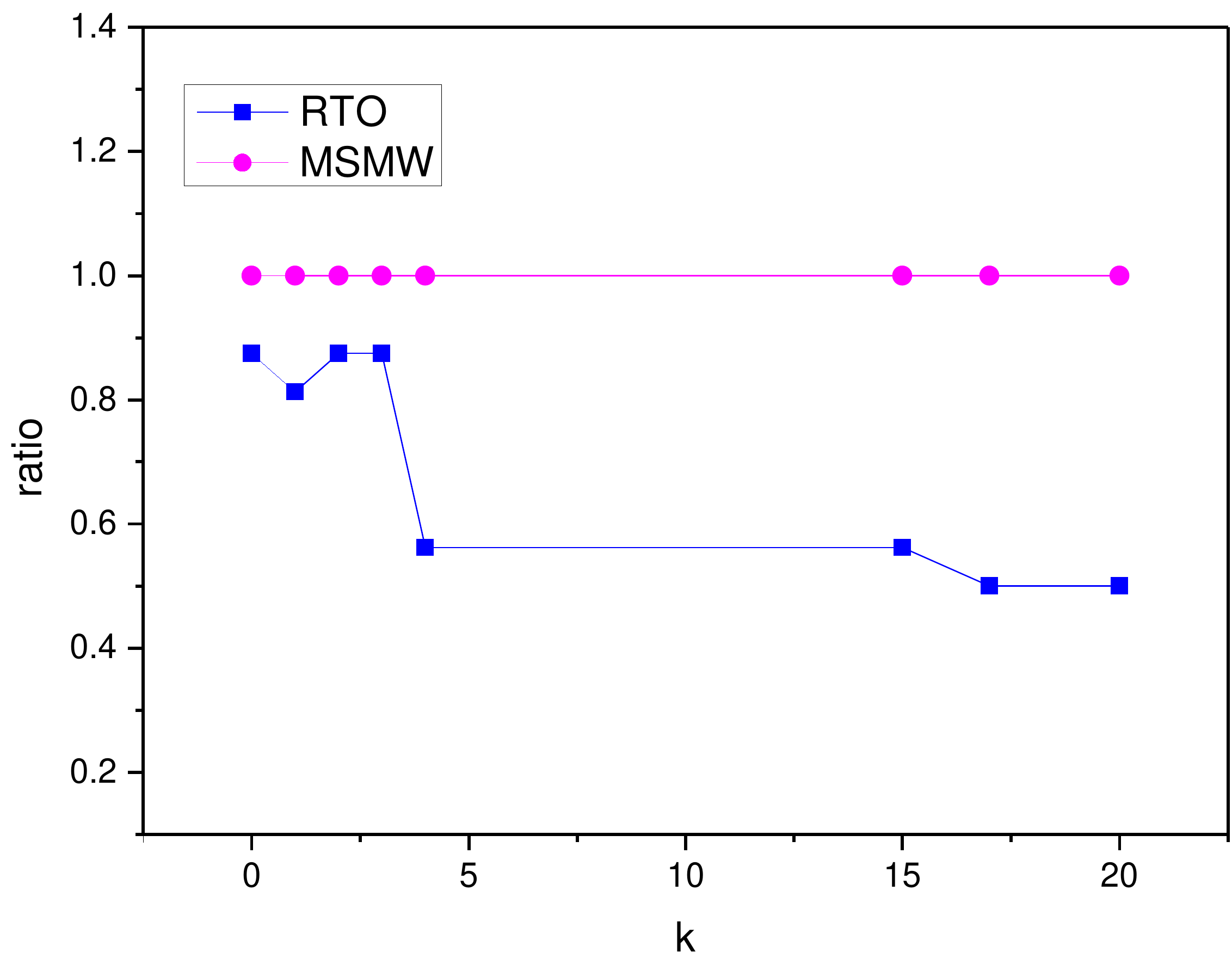}

}
\subfloat[Queue Length]{
\label{fig:improved_subfig_b}
\centering
\includegraphics[scale=0.16]{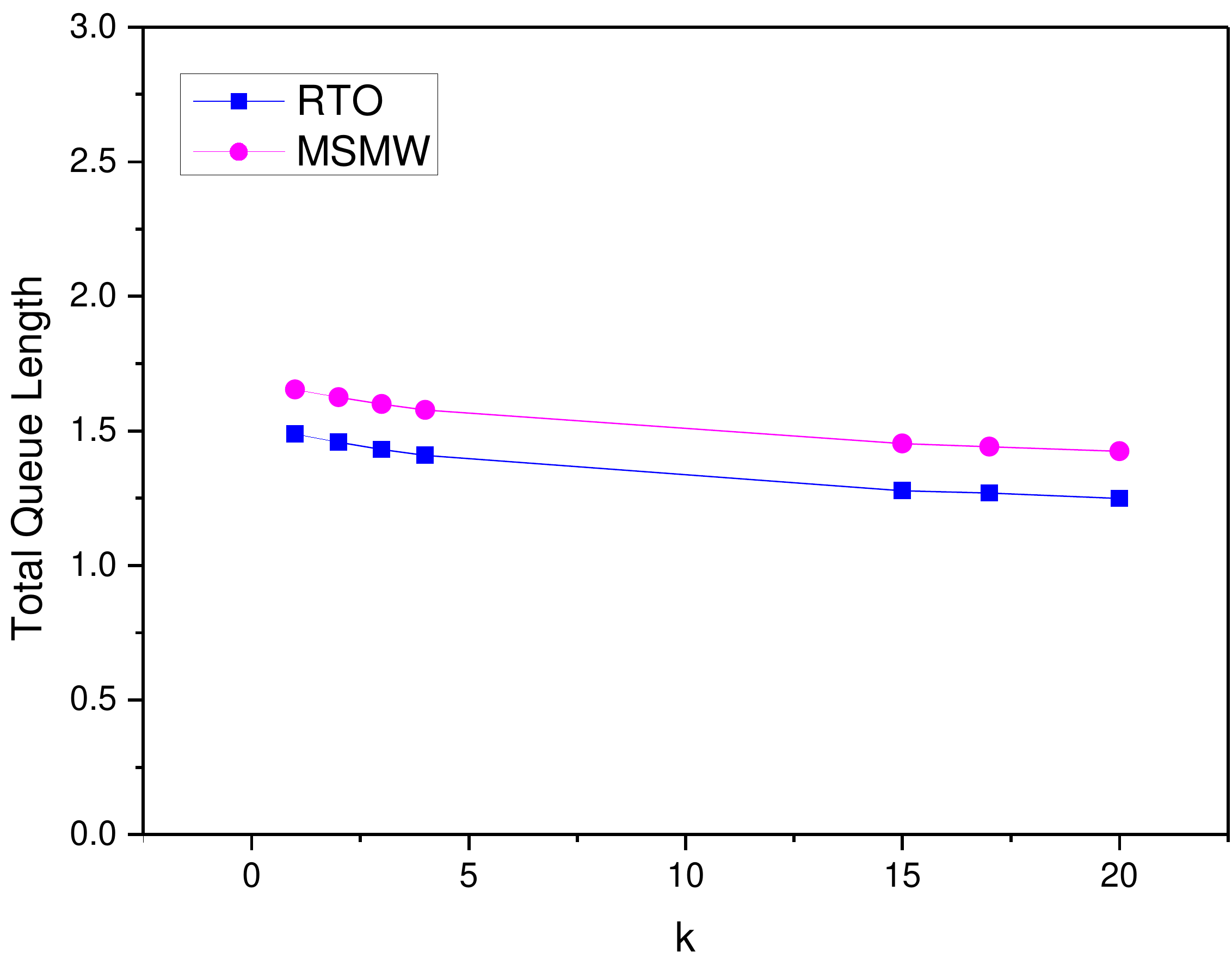}

}
\caption{System performance with various arriving rates}
\end{figure}

\subsection{Service Regularity}
For service regularity, we use the standard deviation of the inter-service time as our measurement.
The inter-service time is the duration between two consecutive transmissions of a link.
A smaller standard deviation of inter-service time indicates a link will approximately wait for a same time between every two consecutive services.
It means that the link has a good performance on service regularity.

There are $16$ links in our simulations, with service frequency constraint $N+5$.
The arriving rates are set as follows: $E(A_i)=\frac{1}{N*(i+k)}$ ($k = 1,2,\cdots,20$) for link $i$.
The result is shown in Fig.5.

As can be seen, while the maximum standard deviations of MSMW and RTO are nearly identical, the average standard deviation of MSMW is much better.
It is only about $\frac{1}{3}$ of that of RTO.
MSMW has a better performance on the aspect of service regularity.
Such an improvement is because of the fact that the service frequency constraints are mainly ensured by $ST_i$.
The evolving of $ST_i$ is regular, which leads to a good performance on service regularity.
The standard deviation equals $0$ when $k=20$, which means all the links in MSMW are scheduled in cycles.
This is achieved by the underlying evolving of $ST_i$s and $w_i$s.

\begin{figure}[htbp]
\centering
\subfloat[Max Standard Deviation]{
\label{fig:improved_subfig_a}

\centering
\includegraphics[scale=0.16]{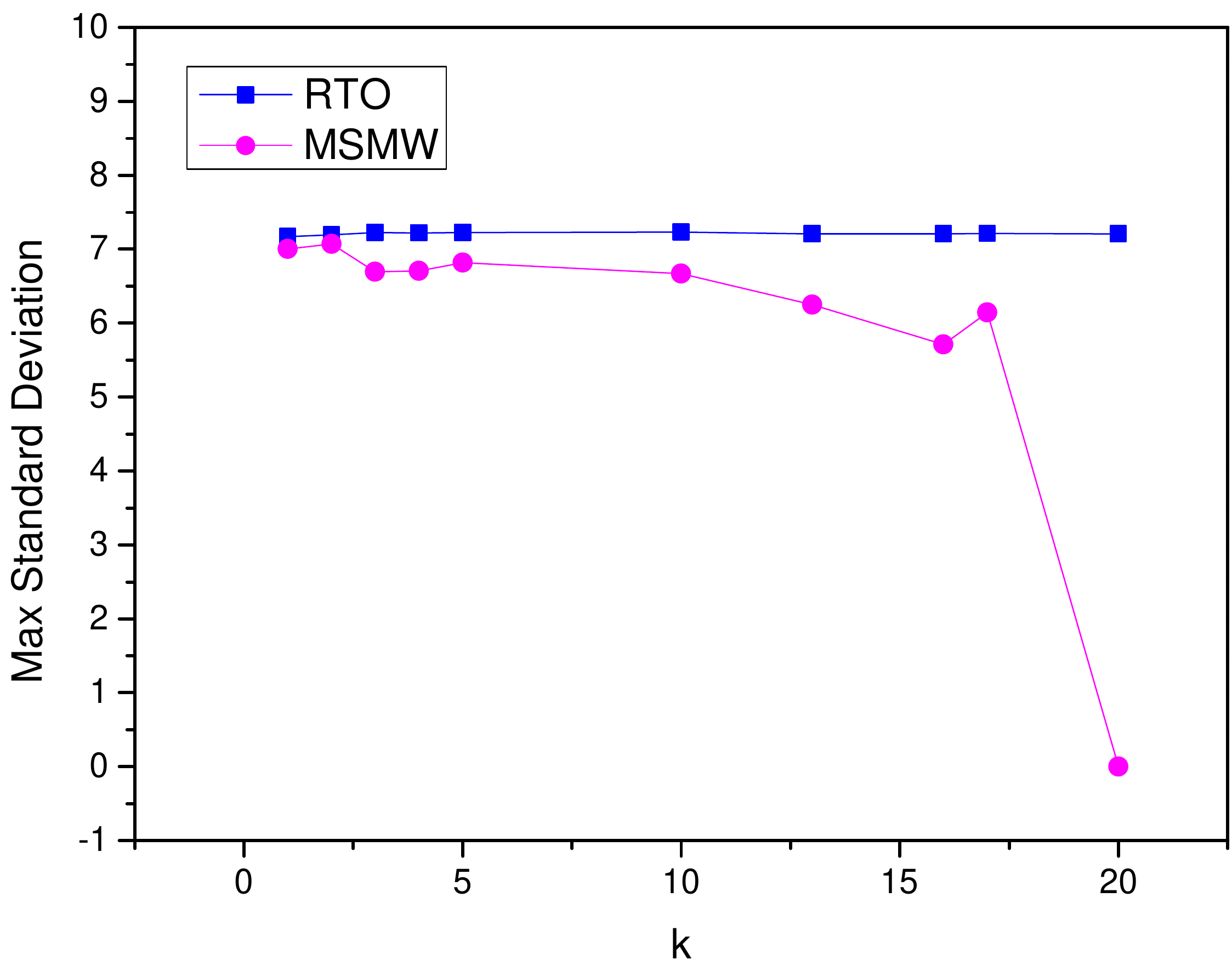}

}
\subfloat[Average Standard Deviation]{
\label{fig:improved_subfig_b}
\centering
\includegraphics[scale=0.16]{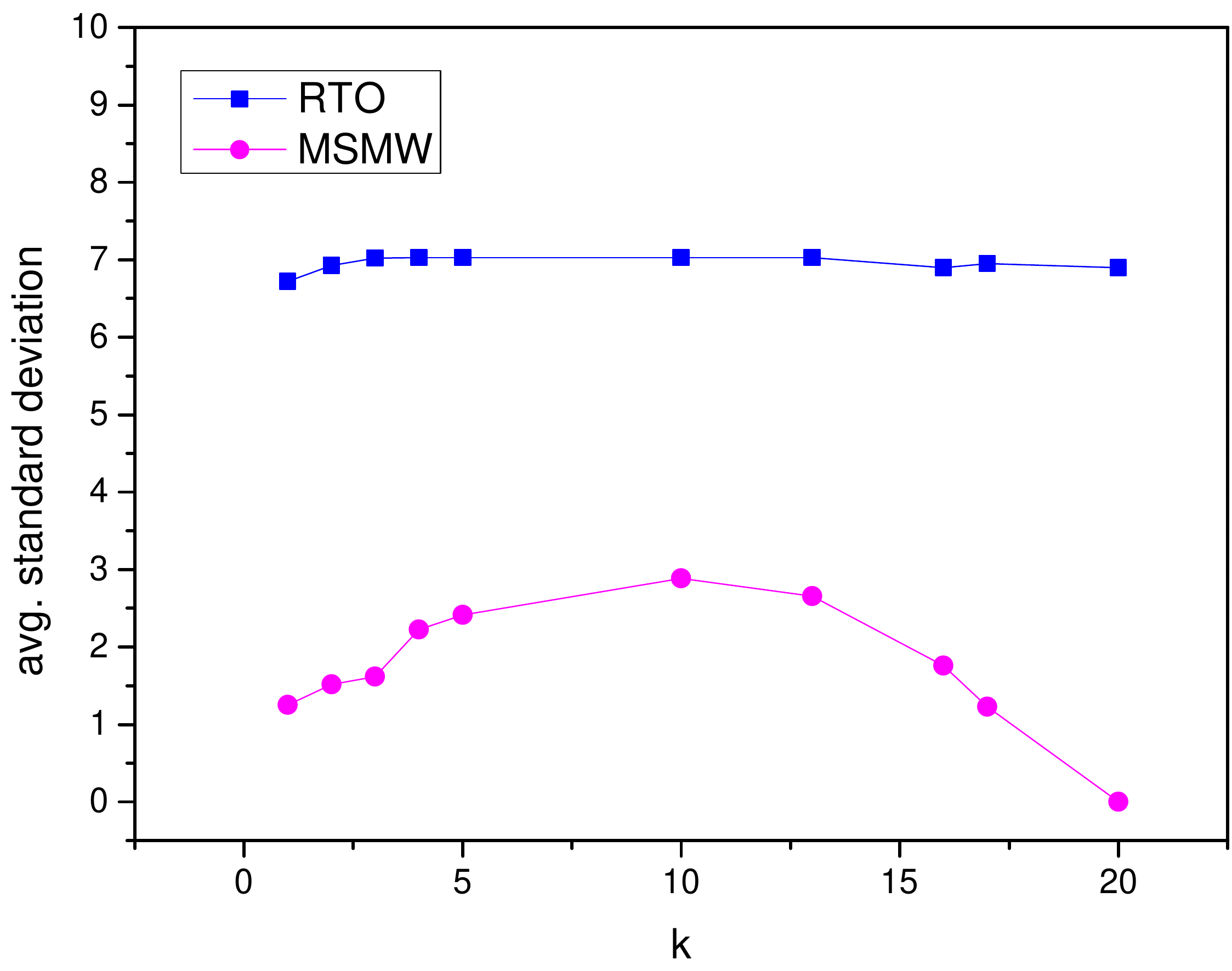}

}
\caption{Performance on service regularity}
\end{figure}

\subsection{Sensitivity for the Arriving Rate Vector}
The performance of MSMW on service frequency does not depend on the arriving rate vectors.
It always has a good performance on the service frequency no matter what the arriving rate vector is.
There are $16$ links in our simulations, with $\delta_i=N+5$ for link $i$.
We use an \textit{asymptotically decreasing} rate vector with $E(A_i)=\frac{1}{N*(i+k)}$ ($k = 1,2,\cdots,20$) and also a \textit{two tails} rate vector with $E(A_i)=\frac{1}{N*i}$ for the first $8$ links, and $E(A_i)=\frac{1}{N*(i+k)}$ ($k = 1,2,\cdots,20$) for the last $8$ links.
The ratio of the links following the service frequency constraints is shown in Fig.6.

As can be seen, MSMW performs well for both types of rate vectors, while the chosen parameters of RTO can only fit the \textit{asymptotically decreasing} rate vector.
We must update the parameters of RTO to deal with the \textit{two tails} one.
Therefore, our scheduling policy is much more applicable in the practical case where the arriving rate of each link is often periodically changing and unpredictable.

\begin{figure}
\centering
\includegraphics[scale=0.185]{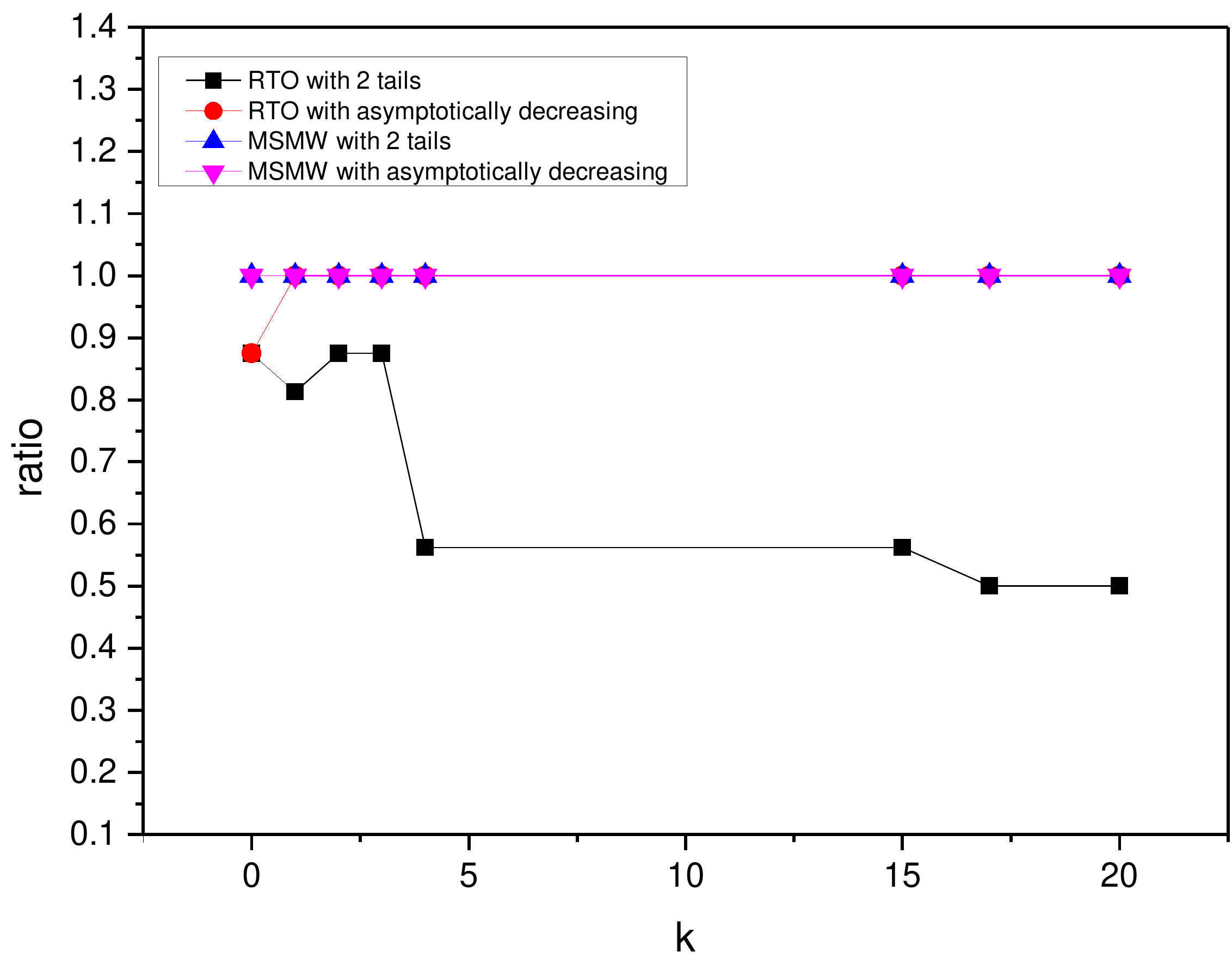}
\caption{Robust of MSMW.}
\label{fig:example}
\end{figure}

\subsection{Performance in General Networks}
In this part, we evaluates the performance of our scheduling policy in general networks.
We assume 40 nodes are randomly distributed in a 100m$\times$100m area.
The transmission and interference radius of a node is set to be 20m.
The channel condition of each link is set to be 0.5,
The transmission rate for each link is 2.
The average arriving rate is set to be $\frac{1}{8}$ and max arriving rate is $\frac{1}{4}$.
We use the maximal weight scheduling policy \cite{gupta2010delay} as our comparison.
In each slot, the maximal weight scheduling policy iteratively selects links with max queue length and minimum degree in the conflict graph, and deletes all links conflicting with the chosen one.

Fig.7 shows the results with different number of links.
The service constraints is set to be 10, and we change the number of links from 20 to 200.
As we see, MSMW has a good performance on both queue length and service frequency.
All the links follow the constraint on service frequency, while the maximal weight scheduling policy can only follow the constraints on 70 percents of its periods.
The average queue length in MSMW is less than 4 times of the that in maximal weight scheduling policy in worst case, \textit{i.e.}, 40/11.5 when there are 160 links.

\begin{figure}[htbp]
\centering
\subfloat[Ratio of Links Achieving Service Frequency]{
\label{fig:improved_subfig_a}

\centering
\includegraphics[scale=0.16]{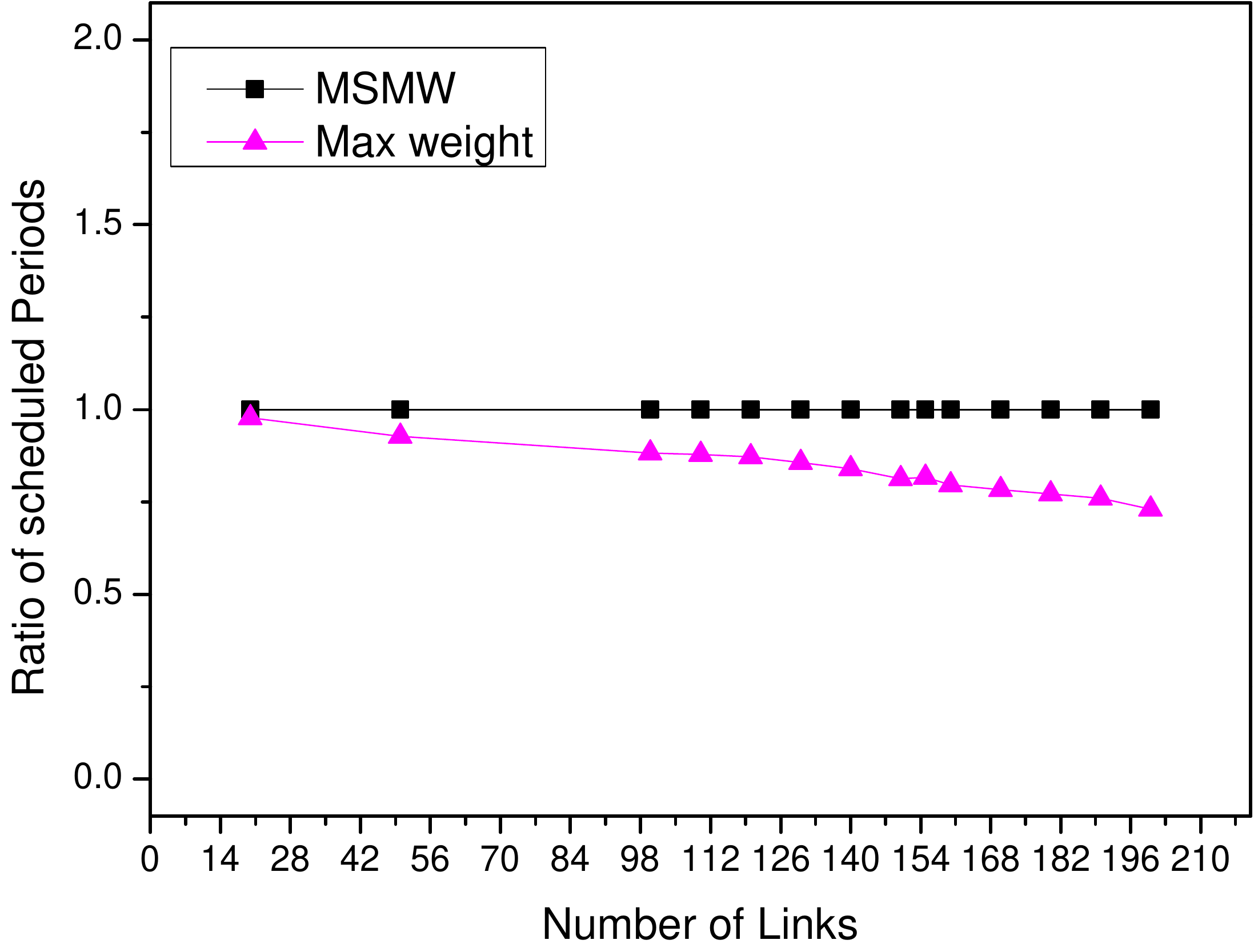}

}
\subfloat[Queue Length]{
\label{fig:improved_subfig_b}
\centering

\includegraphics[scale=0.16]{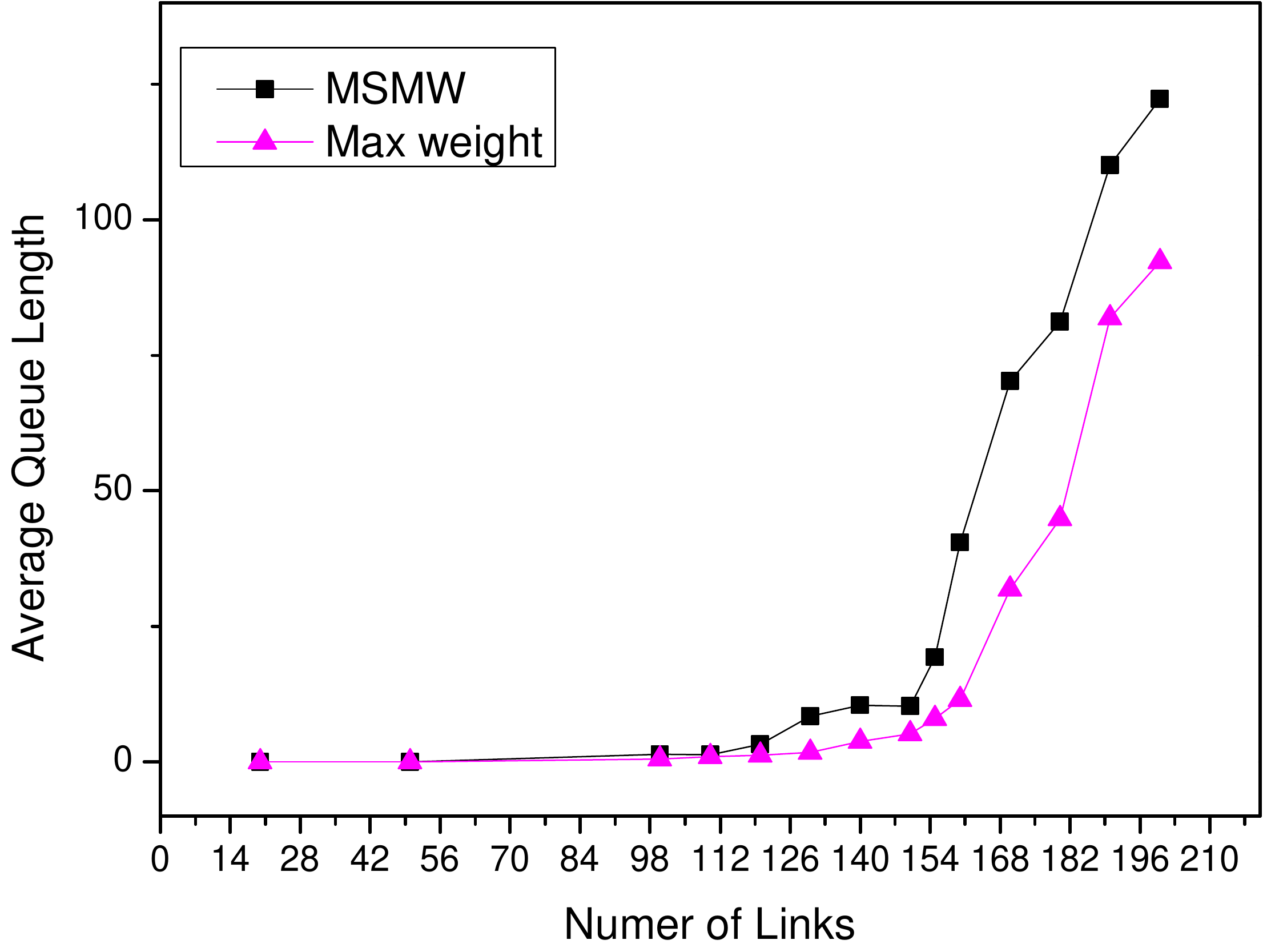}

}
\caption{Performance with different number of links}
\end{figure}

In Fig.8, we evaluate MSMW in general networks with different service frequency  constraints.
We set the number of links to be 160, and change the constraints from 5 to 15.
Based on Fig.8,  the average queue length becomes smaller as $\delta$ increases.
The reason is that more slots can be released for scheduling links according to the queue length.
When $\delta\ge12$, the queue length of MSMW is already small and close to that of maximal weight scheduling policy.
Meanwhile, the performance on service frequency quickly achieve the best in MSMW, while the ratio for our comparison method is only 90 percents even if $\delta=15$.

\begin{figure}[htbp]
\centering
\subfloat[Ratio of Links Achieving Service Frequency]{
\label{fig:improved_subfig_a}

\centering
\includegraphics[scale=0.16]{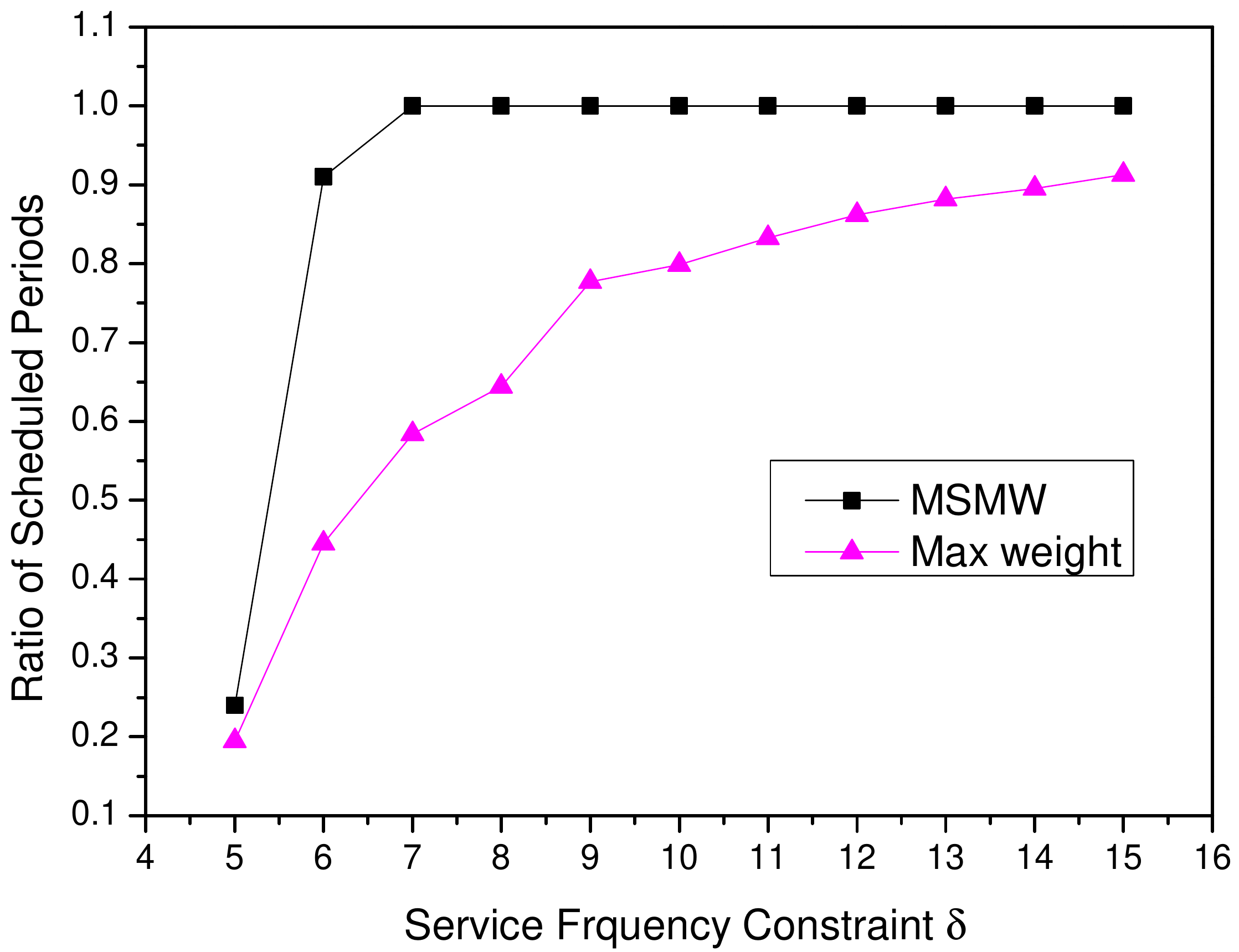}

}
\subfloat[Queue Length]{
\label{fig:improved_subfig_b}
\centering

\includegraphics[scale=0.16]{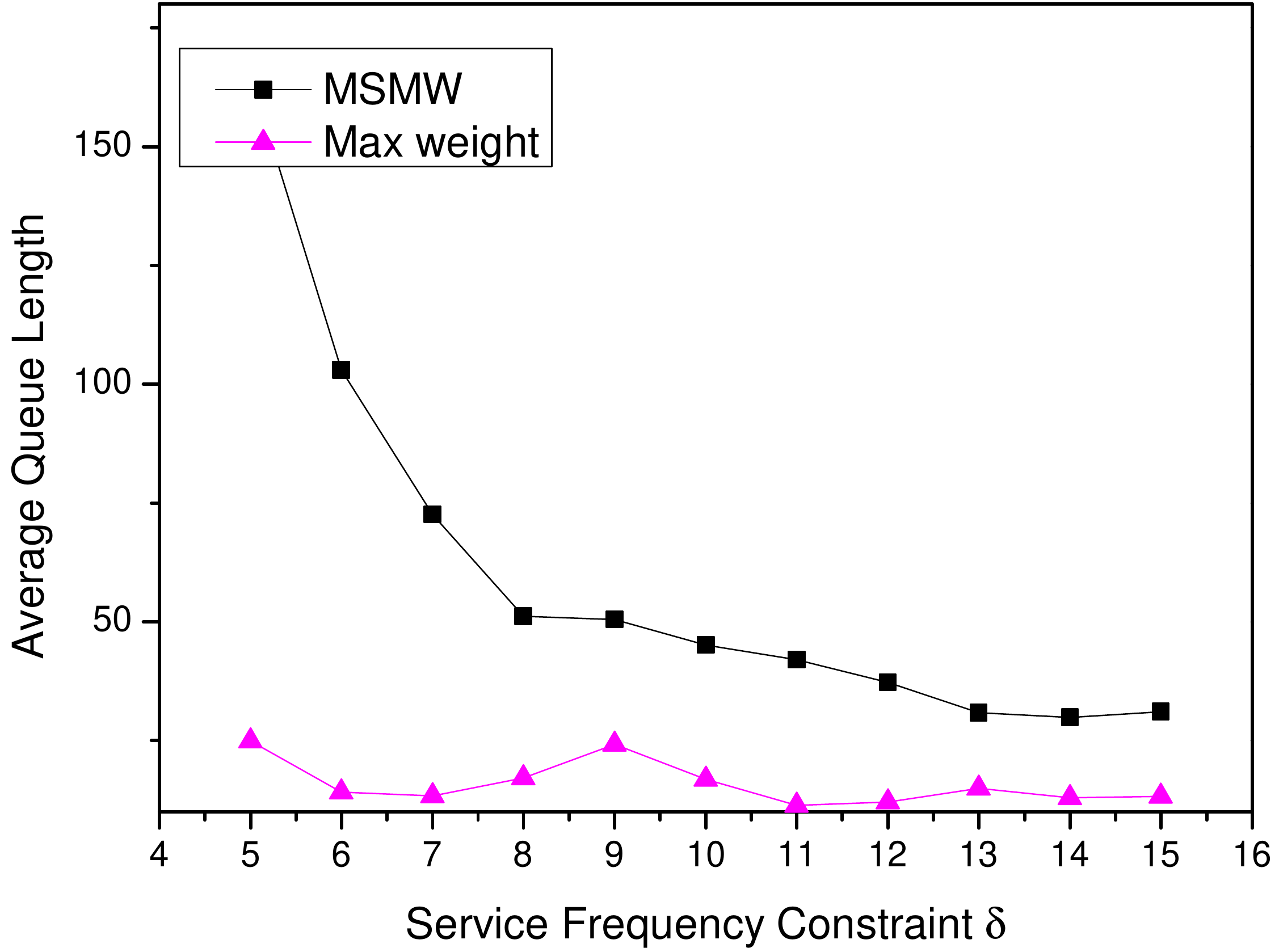}

}
\caption{Performance with different service regularity constraints}
\end{figure}

As a whole, MSMW is efficient no matter what the users' behaviors are.
It always guarantees service frequency while only increases the queue length by a little bit.

\section{Conclusion}
This paper investigates the problem of scheduling links in wireless networks with service frequency constraints.
To the best of knowledge, this is the first work dealing with traffics in wireless networks with stringent service frequency requirements.
We improve the typical network models to handle the constraints.
A new scheduling policy called MSMW is proposed, and it is proved to be throughput-optimal in collocated networks.
Furthermore, this paper discusses the decreasing of capacity region and the changing of queue length caused by the existence of service frequency constraints.
The simulation results show that MSMW performs well on the aspects of service frequency and total queue length.
As future works, we are interested in the relationship between the system performance and the correlation of arriving rates and service frequency constraints.

\bibliographystyle{ieeetr}
\bibliography{Interservicetime}
\section*{Appendix A}
The proof of Lemma 4 is shown here:
\begin{proof}
\textit{Part 1:}
According to Lemma 2, each link $i$ can be scheduled at least once in every $\delta_i$ slots.
Then it can deliver an expectation of $Min\{\frac{1}{\delta_i}, \frac{\lambda_i}{r_ic_i}\}*\delta_i$ packets each time if the link is scheduled only once.
Let $k=1,2,\cdots$.
We find that the expectation queue length of link $i$ at the end of every $k* T$ slots follows
\begin{equation}
 E(Q_i[k* T])\leq Min\{\frac{1}{\delta_i}, \frac{\lambda_i}{r_ic_i}\}*\delta_i+Max\{0,\frac{\lambda_i}{r_ic_i}-\frac{1}{\delta_i}\}*k* T
\end{equation}
where the first part on the right side means the maximum number of packets arrives since link $i$'s last transmission, and the second part is the extra length of the queue if link $i$ has never been scheduled twice in every $\delta_i$ slots.
In (14), the size of the right side approaches to infinite as $k$ increases.

We can further lower the upper bound of $E(Q_i[k*T])$.
We will show that $E(Q_i[k*T])$ is actually bounded despite the increasing of $k$.
We partition the arriving rate (or packets) of link $i$ into two parts:
\textit{round robin} component with $\lambda_{irr}=Min\{\frac{1}{\delta_i}, \frac{\lambda_i}{r_ic_i}\}$, and \textit{maximum weight}  component with $\lambda_{imw}=Max\{0,\frac{\lambda_i}{r_ic_i}-\frac{1}{\delta_i}\}$.

$E(Q_i[k*T])$ is bounded if the \textit{maximum weight} component of link $i$ equals 0, which means link $i$ only needs a single scheduling in every $\delta_i$ slots.

Now consider all the links whose \textit{maximum weight} components are larger than 0.
There are packets in these links that require extra scheduling.
We can generalize the scheduling of \textit{maximum weight} component into a normal maximum weight scheduling if we maintain an extra queue for each link with arriving rate $Max\{A_i[t]-\lambda_i,0\}$ and select the link with the maximum weight based on this queue length each time when we schedule a link in $gp_0$.
This is a typical maximum weight scheduling with $\lambda_{i}'=(\frac{\lambda_i}{r_ic_i}-\frac{1}{\delta_i})/(1-\sum_{i=1}^{N}\frac{1}{\delta_i})$ for link $i$.
According to Lemma 3,
\begin{equation}
\begin{split}
\sum_{i=1}^{N}\lambda_i'&=\sum_{i=1}^{N}Max\{\frac{\lambda_i}{r_ic_i}-\frac{1}{\delta_i},0\}/(1-\sum_{i=1}^{N}\frac{1}{\delta_i})\\
 & < (1-\sum_{i=1}^{N}\frac{1}{\delta_i})=1
\end{split}
\end{equation}
Thus, in the modified system, $\lambda_i'$s are stable and the queue length is bounded by a constant.
By Theorem 2, the expectation of the total queue length is less than $B_1=\frac{1}{2(1-\sum_{j=1}^{N}\frac{1}{\delta_j})^2}\sum_{l_i\in\{l_j|\frac{\lambda_j}{r_jc_j}-\frac{1}{\delta_j}>0\}}[(1-\sum_{j=1}^{N}\frac{1}{\delta_j})(\frac{\lambda_i}{r_ic_i}-\frac{1}{\delta_i})+Var(\frac{A_i[t]}{r_ic_i}-\frac{1}{\delta_i})-(\frac{\lambda_i}{r_ic_i}-\frac{1}{\delta_i})^2]$.
$B_i$ is a constant since $N$, $\delta_i$'s, $\lambda_i$'s, $r_ic_i$'s, and the second moment of $A_i[t]$'s are all bounded.

MSMW processes links according to the total weight, and it schedules a link whose total queue length is the maximum in $gp_0$.
This makes the scheduling different from the modified system.
A link with the maximum queue length in the \textit{maximum weight} component is potentially not the largest when adding up the queue length in the \textit{round robin} component.
This leads to the increasing of the total queue length in the \textit{maximum weight} component when we compare it with the modified system.
However, the link with the maximum queue length in the \textit{maximum weight} component can still be scheduled before the queue length goes to infinity.
This is due to the fact that the queue length in the \textit{round robin} component is bounded in MSMW.
A link with the maximum queue length in the \textit{maximum weight} component will finally grow to the largest one in the total queue length.
Therefore, the total queue length of all the links in the \textit{maximum weight} component can be bounded by $h*B_1$, where $h$ is a constant.

Now we can conclude that $E(\sum_{i=1}^{N}Q_i[k*T])$ is bounded as follows:
\begin{equation}
 E(\sum_{i=1}^{N}Q_i[k*T])\leq \sum_{i=1}^{N}Min\{\frac{1}{\delta_i}, \frac{\lambda_i}{r_ic_i}\}*\delta_i + h*B_1
\end{equation}
\textit{Part 2:}
The second part of the proof is to show that in an arbitrary time slot $t$,
$E(\sum_{i=1}^{N}Q_i[t])$ is bounded.

As we know, the total queue length of all the links usually decreases when MSMW schedules a link in $gp_0$, or the system may be instable.
An exception is the situation when the queue length of all the links are smaller than the transmission rates.
It means the total queue length of all the links is finite, which will not always happen.
Therefore, the largest $E(\sum_{i=1}^{N}Q_i[t])$ often appears when MSMW just scheduled a series of links not in $gp_0$.

For an arbitrary slot $t$, we have $t=k_t*T+t\mod{T}$, where $k_t$ is a non-negative integer.
Then $E(\sum_{i=1}^{N}Q_i[t])$ will be the largest when MSMW schedules the links not in $gp_0$ in all the last $t\mod{T}$ slots.
Based on the design of MSMW, the number of slots for the links not in $gp_0$ is no more than $\sum_{j=1}^{N}\frac{1}{\delta_j}*T$ in a time period of length $T$.
Then we have
\begin{equation}
\begin{split}
E(\sum_{i=1}^{N}Q_i[t]) & \leq E(\sum_{i=1}^{N}Q_i[k_t*T])+\sum_{i=1}^{N}(\sum_{j=1}^{N}\frac{1}{\delta_j}*\frac{\lambda_i}{r_ic_i}*T)-\Gamma\\
&\leq E(\sum_{i=1}^{N}Q_i[k_t*T])+\sum_{i=1}^{N}(\sum_{j=1}^{N}\frac{1}{\delta_j}*\frac{\lambda_i}{r_ic_i}*T)
\end{split}
\end{equation}
where $\Gamma$ relates to the size of packets delivered in $t\mod{T}$ slots.

Finally, we combine Equations (16) and (17),
\begin{equation}
\begin{split}
E(\sum_{i=1}^{N}Q_i[t])\leq &\sum_{i=1}^{N}Min\{\frac{1}{\delta_i}, \frac{\lambda_i}{r_ic_i}\}*\delta_i\\
&+ h*B_1+\sum_{i=1}^{N}(\sum_{j=1}^{N}\frac{1}{\delta_j}*\frac{\lambda_i}{r_ic_i}*T)\\
&\leq N+h*B_1+\sum_{i=1}^{N}(\sum_{j=1}^{N}\frac{1}{\delta_j}*\frac{\lambda_i}{r_ic_i}*T)
\end{split}
\end{equation}
where $B_1=\frac{1}{2(1-\sum_{j=1}^{N}\frac{1}{\delta_j})^2}\sum_{l_i\in\{l_j|\frac{\lambda_j}{r_jc_j}-\frac{1}{\delta_j}>0\}}[(1-\sum_{j=1}^{N}\frac{1}{\delta_j})(\frac{\lambda_i}{r_ic_i}-\frac{1}{\delta_i})+Var(\frac{A_i[t]}{r_ic_i}-\frac{1}{\delta_i})-(\frac{\lambda_i}{r_ic_i}-\frac{1}{\delta_i})^2]$, $N$, $h$, $B_1$, $\{\lambda_i,\delta_i,r_ic_i\}$, and $T$ are constant.
\end{proof}

\section*{Appendix B}
The proof of Theorem 5 employs some idea in Lemma 4.
We consider a new scheduling policy 2-MSMW that is similar to MSMW.
2-MSMW maintains two queues for each link.
There are $A_i[t]$ packets arriving at link $i$ in slot $t$.
2-MSMW adds the first $\frac{r_ic_i}{\delta_i}$ packets to $Q_{i1}$, and the other packets to $Q_{i2}$ if $A_i[t]>\frac{r_ic_i}{\delta_i}$.
In each slot, 2-MSMW maintains $ST_i$ for every link in the same manner as MSMW.
2-MSMW schedules the links based on the weights of $Q_{i1}$ when it selects a link not in $gp_0$, and based on the weights of $Q_{i2}$ when selecting links from $gp_0$.
Furthermore, 2-MSMW only transmits packets in the queue that is used to calculate the weight in that slot.
The total queue length of 2-MSMW is $E(\sum_{i=1}^{N}Q'_i[t])=E(\sum_{i=1}^{N}(Q_{i1}[t]+Q_{i2}[t]))$.

Now we consider the total queue length of MSMW $E(\sum_{i=1}^{N}Q_i[t])$.	
The stages of each link in the two scheduling policies are identical.
In each time slot, MSMW calculates the weight of each link with $Q_{i1}+Q_{i2}$, and we have $Max\{Q_{i1}+Q_{i2}\}\geq Max\{Q_{i1},Q_{i2}\}$.
Therefore, the number of the packets that departure from the system in MSMW is no less than the number of the packets that are transmitted in 2-MSMW.
The numbers are the same only when 2-MSMW makes full use of its channel resource.
Furthermore, the arriving rate vectors of them are also the same.
Then we simply add up the transmitted packets in each slot and conclude our result as follows:
$E(\sum_{i=1}^{N}Q_i[t])\leq E(\sum_{i=1}^{N}Q'_i[t])$.

The evolving of $Q_{i1}$ follows a round-robin-like manner where link $i$ transmits once in every $\delta_i$ slots.
This also means the evolving of $Q_{i1}$'s is periodical when the system has been carrying out 2-MSMW for a sufficient long time and turned into a steady state.
Then
\begin{equation}
\begin{split}
&\text{\ \ \ }E(\lim_{t\rightarrow \infty}\sum_{t'=1}^{t}\sum_{i=1}^{N}Q_{i1}[t'])=E(\lim_{t\rightarrow \infty}\sum_{i=1}^{N}\sum_{t'=1}^{t}Q_{i1}[t'])\\
&=\lim_{t\rightarrow \infty}\sum_{i=1}^{N}Min\{\frac{\lambda_i}{r_ic_i},\frac{1}{\delta_i}\}*(1+2+\cdots+\delta_i)*\frac{t}{\delta_i}\\
&=\lim_{t\rightarrow \infty}\sum_{i=1}^{N}Min\{\frac{\lambda_i}{r_ic_i},\frac{1}{\delta_i}\}*\frac{\delta_i(1+\delta_i)}{2}*\frac{t}{\delta_i}\\
&=\lim_{t\rightarrow \infty}\sum_{i=1}^{N}Min\{\frac{\lambda_i}{r_ic_i},\frac{1}{\delta_i}\}*\frac{(1+\delta_i)*t}{2}
\end{split}
\end{equation}
$E(Q_{i1})$ has been estimated in Lemma 4.
Since 2-MSMW processes $Q_{i1}$ and $Q_{i2}$ separately, we have
\begin{equation}
\begin{split}
&\text{\ \ \ }E(\sum_{i=1}^{N}Q'_i[t])=E(\sum_{i=1}^{N}(Q_{i1}[t]+Q_{i2}[t]))\\
&=\lim_{t\rightarrow \infty}\frac{\sum_{t'=1}^{t}\sum_{i=1}^{N}Q_{i1}[t']}{t}+\lim_{t\rightarrow \infty}\frac{\sum_{t'=1}^{t}\sum_{i=1}^{N}Q_{i2}[t']}{t}\\
&=\lim_{t\rightarrow \infty}\frac{\sum_{i=1}^{N}\sum_{t'=1}^{t}Q_{i1}[t']}{t}+E(\sum_{i=1}^{N}Q_{i2}[t])\\
&\leq \sum_{i=1}^{N}Min\{\frac{\lambda_i}{r_ic_i},\frac{1}{\delta_i}\}*\frac{(1+\delta_i)}{2}+B_1\\
\end{split}
\end{equation}
Finally,

\begin{equation}
\begin{split}
E(\sum_{i=1}^{N}Q_i[t])&\leq \sum_{i=1}^{N}(Min\{\frac{\lambda_i}{r_ic_i},\frac{1}{\delta_i}\}*\frac{(1+\delta_i)}{2})\\
&+\frac{1}{2(1-\sum_{i=1}^{N}\frac{1}{\delta_i})^2}\sum_{l_i\in\{l_j|\frac{\lambda_j}{r_jc_j}-\frac{1}{\delta_j}>0\}}[(1-\sum_{k=1}^{N}\frac{1}{\delta_k})\\
&*(\frac{\lambda_i}{r_ic_i}-\frac{1}{\delta_i})+Var(\frac{A_i[t]}{r_ic_i}-\frac{1}{\delta_i})-(\frac{\lambda_i}{r_ic_i}-\frac{1}{\delta_i})^2]\\
&\leq \frac{1}{2}+\sum_{i=1}^{N}\frac{1}{\delta_i}*\frac{\delta_i}{2}+\frac{1}{2\epsilon^2}\sum_{l_i\in\{l_j|\frac{\lambda_j}{r_jc_j}-\frac{1}{\delta_j}>0\}}\\
&[\epsilon*(\frac{\lambda_i}{r_ic_i}-\frac{1}{\delta_i})-(\frac{\lambda_i}{r_ic_i}-\frac{1}{\delta_i})^2+\frac{1}{(r_ic_i)^2}Var(A_i[t])]\\
&\leq \frac{N+1}{2}+\frac{1}{2\epsilon^2}\sum_{l_i\in L'}[\epsilon*\frac{\epsilon}{2}-(\frac{\epsilon}{2})^2]\\
&+\frac{1}{2\epsilon^2}\sum_{l_i\in L'}\frac{Var(A_i[t])}{(r_ic_i)^2}\\
&\leq \frac{N+1}{2}+\frac{N'}{8}+\frac{1}{2\epsilon^2}\sum_{l_i\in L'}\frac{Var(A_i[t])}{(r_ic_i)^2}
\end{split}
\end{equation}

When $\epsilon = 0$, the upper bound of the queue length can be directly derived from (14) as follows:
\begin{equation}
\lim_{t\rightarrow \infty}\frac{\sum_{t'=1}^{t}\sum_{i=1}^{N}Q_i[t']}{t}\leq \sum_{i=1}^{N}(Min\{\frac{\lambda_i}{r_ic_i},\frac{1}{\delta_i}\}*\frac{(1+\delta_i)}{2})
\end{equation}

\end{document}